\documentclass[a4paper]{article}

\usepackage[hidelinks]{hyperref}
\usepackage{siunitx}
\usepackage{graphicx}
\usepackage{amssymb}
\usepackage{amsmath}
\usepackage{amsthm}
\usepackage{mathabx}
\usepackage{bm}
\usepackage{epstopdf}
\usepackage{algorithm}
\usepackage[noend]{algpseudocode}
\usepackage{tikz}
\usetikzlibrary{shapes.geometric}
\usepackage{natbib}
\usepackage{adjustbox}
\usetikzlibrary{calc}
\usepackage[title]{appendix}

\date{}
\begin{document}

\title{Parallelising Particle Filters with Butterfly Interactions}


\author{Kari Heine\footnote{Corresponding author, email: k.m.p.heine@bath.ac.uk, Department of Mathematical Sciences, University of Bath, Bath, BA2 7AY, UK}, Nick Whiteley\footnote{School of Mathematics, University of Bristol, Bristol, BS8 1TW, UK}, A. Taylan Cemgil\footnote{Deptartment of Computer Engineering, Bo{\u g}azi{\c c}i University, Istanbul, 34342 Bebek, Turkey}}

%






\newcommand{\aMatrices}{\ast}
\newcommand{\assref}[1]{\ref{1}}
\newcommand{\bA}{\overline{A}}
\newcommand{\barphi}{\overline{\varphi}}
\newcommand{\bGam}{\overline{\Gamma}}
\newcommand{\bN}{\mathbb{N}}
\newcommand{\boundMeas}{\mathcal{B}}
\newcommand{\bcF}{\overline{\mathcal{F}}}
\newcommand{\bQ}{\overline{Q}}
\newcommand{\bV}{\overline{V}}
\newcommand{\bX}{\overline{X}}
\newcommand{\bxi}{\boldsymbol{\xi}}
\newcommand{\cA}{\mathcal{A}}
\newcommand{\cF}{\mathcal{F}}
\newcommand{\cM}{\mathcal{M}}
\newcommand{\cN}{\mathcal{N}}
\newcommand{\cX}{\mathcal{X}}
\newcommand{\cY}{\mathcal{Y}}
\newcommand{\defeq}{:=}
\newcommand{\E}{\mathbb{E}}
\newcommand{\iidsim}{\stackrel{\mathrm{iid}}{\thicksim}}
\newcommand{\infnorm}[1]{\left\|#1\right\|}
\newcommand{\Id}{\mathbf{I}}
\newcommand{\id}{Id}
\newcommand{\IN}{\mathrm{in}}
\newcommand{\ind}{\mathbb{I}}
\renewcommand{\N}{\mathcal{N}}
\renewcommand{\O}{\mathcal{O}}
\newcommand{\Osc}[1]{\mathrm{Osc}\left(#1\right)}
\newcommand{\osc}[1]{\mathrm{osc}\left(#1\right)}
\newcommand{\ones}{\mathbf{1}}
\newcommand{\one}{\mathbf{1}}
\newcommand{\OUT}{\mathrm{out}}
\renewcommand{\P}{\mathbb{P}}
\newcommand{\R}{\mathbb{R}}
\renewcommand{\ss}{\mathbb{X}}
\newcommand{\tF}{\widetilde{\mathcal{F}}}
\newcommand{\tilphi}{\widetilde{\varphi}}
\newcommand{\tvnorm}[1]{\left\|#1\right\|}
\newcommand{\tX}{\widetilde{X}}
\newcommand{\ud}{\mathrm{d}}
\newcommand{\vbarphi}{\overline{\boldsymbol{\varphi}}}
\newcommand{\vg}{\boldsymbol{g}}
\newcommand{\vphi}{\boldsymbol{\varphi}}
\newcommand{\vW}{\boldsymbol{W}}
\newcommand{\vx}{\boldsymbol{x}}
\newcommand{\vy}{\boldsymbol{y}}
\newcommand{\vxi}{\boldsymbol{\xi}}
\newcommand{\vzeta}{\boldsymbol{\zeta}}
\newcommand{\wcF}{\widecheck{\mathcal{F}}}
\newcommand{\wcX}{\widecheck{X}}
\newcommand{\cV}{\mathcal{V}}
\newcommand{\X}{\mathbb{X}}
\newcommand{\Y}{\mathbb{Y}}
\newcommand{\zeros}{\mathbf{0}}

\newcommand{\lpnorm}[2]{\E\left[\left|#2\right|^{#1}\right]^{\frac{1}{#1}}}
\newcommand{\lp}[2]{\E[|#2|^{#1}]^{1/{#1}}}
\newcommand{\lpNorm}[2]{\left\|#2\right\|_{#1}}
\newcommand{\floor}[1]{\left\lfloor#1\right\rfloor}

\newcommand{\almostsurelyOneArg}[1]{\xrightarrow[#1]{\mathrm{a.s.}}}

\newcommand{\note}[1]{{\color{red}\textbf{(#1)}}}

\newtheorem{assumption}{Assumption}
\newtheorem{proposition}{Proposition}
\newtheorem{lemma}{Lemma}
\newtheorem{theorem}{Theorem}
\renewcommand{\theassumption}{A\arabic{assumption}}
\def\proofs{0}
\renewcommand{\theenumi}{\alph{enumi}}
\theoremstyle{definition}
\newtheorem{remark}{Remark}

\maketitle

\begin{abstract}
Bootstrap particle filter (BPF) is the corner stone of many popular algorithms used for solving inference problems involving time series that are observed through noisy measurements in a non-linear and non-Gaussian context. The long term stability of BPF arises from particle interactions which in the context of modern parallel computing systems typically means that particle information needs to be communicated between processing elements, which makes parallel implementation of BPF nontrivial.

In this paper we show that it is possible to constrain the interactions in a way which, under some assumptions, enables the reduction of the cost of communicating the particle information while still preserving the consistency and the long term stability of the BPF. Numerical experiments demonstrate that although the imposed constraints introduce additional error, the proposed method shows potential to be the method of choice in certain settings.

\end{abstract}

\section{Introduction}

In modern computing systems an increase in the computational power is primarily obtained by increasing the number of parallel processing elements (PE) rather than by increasing the speed (i.e.~the clock rate) of an individual PE (see e.g.~\cite{pacheco2011}). While in many cases such parallel systems have enabled the completion of increasingly complex computational tasks, they can only do so if the task in question admits parallel computations. In this paper we focus on an important class of algorithms lacking such inherent parallelism, namely the sequential Monte Carlo (SMC) methods, or particle filters~\citep{gordon_et_al_93,doucet_et_al01}. 

It is well known~\citep{lee_et_al10,murray_et_al16} that the complications in parallelising SMC methods are due to the same key ingredient that also underpins their popularity: particle interactions, also commonly referred to as resampling. While these interactions stabilise the algorithms in time, and under certain assumptions, enable time uniform approximations~(see, e.g.~\cite{delmoral_et_guionnet01, douc_et_al14}), they also imply that in an attempt to speed up the computations by distributing the particles across a number of PEs, we will inevitably introduce some \emph{communication cost}. This cost arises from the need to communicate the particle information between PEs to enable the interaction. In this paper we propose new SMC algorithms that are based on an underlying principle of constraining the particle interactions in a structured way with the aim of reducing the communication cost. The resulting algorithms are studied both theoretically and in practice. 

Our theoretical study involves analysing the convergence of the algorithms in the mean of order $r\geq 1$. More specifically, we obtain convergence rates in two specific scenarios:
\begin{enumerate}
\item $m$ is fixed and $M\to\infty$,
\item $m \to \infty$ and $M$ is fixed,
\end{enumerate}
where $m$ denotes the number of PEs and $M$ denotes the number particles per PE. In the former case, the proposed algorithms retain the standard  Monte Carlo rate $M^{-1/2}$ of convergence, while in the latter case a lower $(\log_2(m)/m)^{1/2}$ rate is obtained. 

For the practical study, we compare some of the proposed algorithms empirically in a parallel computation context to a previously proposed SMC algorithm known as the \emph{island particle filter} (IPF) \citep{verge_et_al15} which we regard as the state of the art methodological approach to parallelising SMC. In this paper, we focus on methodology and hence further discussion on more implementation focused approaches, such as those discussed in \cite{murray_et_al16}, is omitted.

Although the numerical experiments may leave some room for speculation on the optimality of the tested implementations, the proposed methods have two specific properties that can be used to introduce gain in performance as demonstrated by the experiments: they enable a more flexible adaptive resampling scheme --- completely unique to the proposed approach --- and they allow a straightforward way of reducing the cost of communicating the particle information between PEs.

\subsection{Particle filters and parallelising them}

The well-known \emph{bootstrap particle filter} (BPF), introduced by \citet{gordon_et_al_93}, first simulates an independent and identically distributed (i.i.d.) sample $\vzeta_{0} \defeq (\zeta^{1}_{0},\cdots,\zeta^{N}_{0})$ from a distribution $\pi_{0}$ defined on a sufficiently regular measurable state space $(\X,\cX)$. Then, for each $n>0$, BPF subsequently generates samples $\vzeta_{n}\defeq (\zeta^{1}_{n},\cdots,\zeta^{N}_{n})$ according to 
\begin{align*}
\zeta^{i}_{n} \iidsim \frac{\sum_{j=1}^{N} g(\zeta^{j}_{n-1},y_{n-1})f(\zeta^{j}_{n-1},\,\cdot\,)}{\sum_{j=1}^{N} g(\zeta^{j}_{n-1},y_{n-1})},\qquad 1 \leq i \leq N,
\end{align*}
where $f:(\X,\cX) \to [0,1]$ is a Markov kernel, and for all $x\in\X$ and some Markov kernel $G:(\X,\cY)\to[0,1]$, the function $g(x,\,\cdot\,)$ is a density of $G(x,\,\cdot\,)$ w.r.t.~some $\sigma$-finite measure on the measurable space $(\Y,\cY)$. The samples $(\vzeta_{n})_{0\leq n}$ then define empirical probability measures
\begin{align*}
\pi^{N}_{n} \defeq \frac{1}{N}\sum_{i=1}^{N}\delta_{\zeta^{i}_{n}}, \qquad n \geq 0,
\end{align*}
where $\delta_{x}$ denotes a point mass located at $x \in \X$. Many convergence results and central limit theorems exist for these measures, see e.g.~\citep{crisan_et_doucet02,del1999central,smc:the:C04,delmoral04}, and it is well known that the limiting distribution of $\pi^{N}_{n}$ is the prediction distribution 
\begin{align*}
\pi_{n}(\,\cdot\,) \defeq \P(X_{n} \in \,\cdot\,\,\mid\,Y_{0}=y_{0},\ldots,Y_{n-1} = y_{n-1}),
\end{align*}
where $X \defeq (X_{n})_{n\geq 0}$ and $Y \defeq (Y_{n})_{n\geq 0}$ are the $\X$ valued signal process and $\Y$ valued observation process, respectively, of the hidden Markov model (HMM)
\begin{equation}\label{eq:HMM}
\arraycolsep=1.4pt
\begin{array}{rll}
X_0 \thicksim \pi_{0},\quad X_{n}\mid X_{n-1}=x_{n-1} &\thicksim f(x_{n-1},\,\cdot\,)& \qquad n\geq 1, \\[.2cm]
 Y_{n} \mid X_{n} = x_{n} &\thicksim g(x_{n},\,\cdot\,) & \qquad n \geq 0. 
\end{array}
\end{equation}

BPF can be summarised as shown in Algorithms \ref{alg:bpf} and \ref{alg:multinomial_resampling}, where we have also used the notations $\widehat{\vzeta}_{n} \defeq (\widehat{\zeta}^{1}_{n},\ldots,\widehat{\zeta}^{N}_{n})$ and $g_{n}(\,\cdot\,) \defeq g(\,\cdot\,,y_n)$ for all $n\geq 0$. We assume that $g_{n}$ is a strictly positive, bounded and measurable function defined in $\X$. The final loop on lines \ref{bpf:mutation loop} and \ref{bpf:mutation} of Algorithm \ref{alg:bpf} we refer to as \emph{the mutation step}.
\begin{algorithm}[!ht]
\caption{Particle filter}
\label{alg:bpf}
\begin{algorithmic}[1]
\For{$i=1,\ldots,N$} 
	\State $\zeta^i_0 \iidsim \pi_{0}$
\EndFor
\For{$n\geq 0$} 

	\State $\widehat{\vzeta}_{n} \leftarrow $\textsc{Resample}$(\vzeta_{n},g_{n})$
	\For{$i=1,\ldots,N$} \label{bpf:mutation loop}
		\State $\zeta^{i}_{n+1} \sim f(\widehat{\zeta}^{i}_{n},\,\cdot\,)$\label{bpf:mutation}
	\EndFor
\EndFor
\end{algorithmic}
\end{algorithm}
\vspace{-4mm}
\begin{algorithm}[!ht]
\caption{Multinomial resampling}
\label{alg:multinomial_resampling}
\begin{algorithmic}[1]
\label{alg:mr}
\State $(\xi_{1}^i)_{1\leq i\leq N}=\textsc{Resample}\left((\xi_{0}^i)_{1\leq i\leq N},g\right)$
	\For{$i=1,\ldots,N$} 
		\State $\xi^{i}_{1} \sim {\sum_{j=1}^{N} g(\xi^{j}_{0})\delta_{\xi^{i}_{0}}}/{\sum_{j=1}^{N} g(\xi^{j}_{0})}$ \label{line:bpf:resampling}
	\EndFor
\end{algorithmic}
\end{algorithm}

An obvious starting point for designing parallel SMC algorithms is to assign $M$ particles to $m$ PEs making the total sample size $N = mM$. Most of the calculations in Algorithms \ref{alg:bpf} and \ref{alg:multinomial_resampling} can be done straightforwardly in parallel, except for line \ref{line:bpf:resampling} in Algorithm \ref{alg:multinomial_resampling} where $\xi^{i}_{1}$ is generated as a duplicate of a random element of $(\xi^{1}_{0},\ldots,\xi^{N}_{0})$. Due to this step, PEs cannot proceed independently, but are required to exchange information about the particle coordinates $\xi^{i}_{0}$ and their associated weights $g(\xi^{i}_{0})$. In this paper we propose new ways of performing this interaction in order to harness the power of parallel computation for more efficient particle filter algorithms.

One of the most important earlier contributions to the design of parallel SMC algorithms is \citep{bolic_et_al_05} which introduced a modification of the BPF whereby the particle interactions are constrained  by allowing the $m$ PEs to exchange subsets of particles according specific local schemes. The theoretical properties of these popular local exchange particle filters (LEPF) was further investigated in \citep{miguez07, miguez14, miguez_et_vazquez15, heine_et_whiteley16}. The analysis of \cite{miguez14, miguez_et_vazquez15} proved that under specific assumptions, the LEPF was uniformly convergent in time as $m\to\infty$, but interestingly, in addition to the central limit theorem for the LEPF, it was shown in \citep{heine_et_whiteley16} that under some regularity assumptions, LEPF cannot be uniformly convergent in mean of order $r\geq 1$ at rate $m^{-1/2}$. Whether the time uniform convergence holds at any slower rate remains an open question. Although the present paper does not address this question directly, it sheds some light on the matter as we show that particle interactions can indeed be constrained in a manner which preserves the time uniform convergence at a slower rate.

A more recent development towards parallelising particle filters is the island particle filter (IPF) proposed by \citet{verge_et_al15}. IPF is based on a two stage implementation of the resampling step. At the first stage one resamples the particle islands, or PEs, to duplicate and redistribute the PE specific particle sets according to some, e.g.~multinomial, resampling scheme without considering particles individually. At the second stage, each PE then performs particle level  resampling independent of each other. \citet{delmoral17} provides proofs of convergence in probability, central limit theorem and large deviations for the IPF algorithm. The methods we propose in the present work are reminiscent to IPF and can be thought of as a result of combining IPF with concepts originating from computer network topologies.


\subsection{Augmented resampling}
\label{sec:r2pf intro}

The particle filter algorithms presented in this paper are all based on a novel \emph{augmented resampling} algorithm which is a multi-stage resampling algorithm parametrised by two positive integers $N$ and $S$ that are assumed to satisfy:
\begin{assumption}\label{ass:NS}
$N,S \in \{1,2,\ldots\}$ are such that $N = mM$ and $S = \log_{2}(m)$ for some $m, M \in \{1,2,\ldots\}$.
\end{assumption}
\noindent
We retain the interpretation of $m$ being the number of PEs, $M$ the number of particles per PE, and $N$ being the total number of particles. The parameter $S$ is specific to the augmented resampling algorithm and it denotes the number of resampling stages. For given matrices $A_{1},\ldots,A_{S} \in \R^{N\times N}$, to be specified later, augmented resampling proceeds as described in Algorithm \ref{alg:augmented_resampling}.
\begin{algorithm}[!ht]
\caption{Augmented resampling}
\label{alg:augmented_resampling}
\begin{algorithmic}[1]
\State {$(\xi^{i}_{S})_{1\leq i \leq N}=\textsc{AugmentedResample}\left((\xi^{i}_{0})_{1\leq i \leq N},g\right)$}{}
\For{$i=1,\ldots,N$}
	\State $V_{0}^{i} \leftarrow g(\xi_{0}^{i})$
\EndFor
\For{$s=1,\ldots,S$}
	\For{$i=1,\ldots,N$}
		\State $V^i_{s} \leftarrow \sum_{j=1}^{N}A^{ij}_{s}V^j_{s-1}$
		\State $\xi^i_{s}\sim (V^i_{s})^{-1}\sum_{j=1}^{N}A^{ij}_{s}V^j_{s-1}\delta_{\xi^j_{s-1}}$
	\EndFor
\EndFor
\end{algorithmic}
\end{algorithm}


A key characteristic of augmented resampling is that by means of the matrices $A_{1},\ldots,A_{S}$, we can control which PEs are allowed to interact at each stage $1 \leq s \leq S$. While our theory allows for defining these matrices in various ways, we will only focus on a specific definition which implies \emph{pairwise} interactions between PEs at each stage $1 \leq s \leq S$. Formally
\begin{equation}\label{eq:A matrices}
A_{s} \defeq \Id_{2^{S-s}} \otimes \ones_{1/2} \otimes \Id_{2^{s-1}} \otimes \ones_{1/M},
\end{equation}
where $\otimes$ denotes the Kronecker product and for any $k>0$, $\Id_{k}$ is size $k$  identity matrix, and the abusive notation $\ones_{1/k}$ is used for a size $k$ matrix of ones multiplied by $1/k$. 


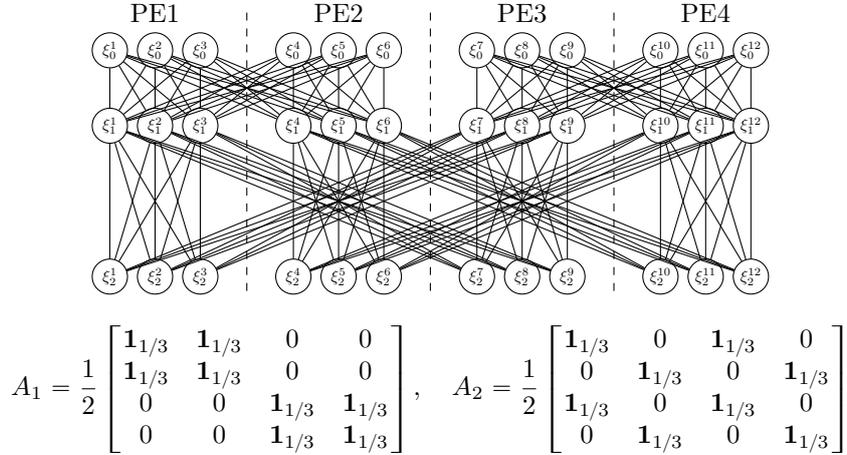
\begin{figure}
\begin{center}
\begin{tabular}{@{}c@{}}
\begin{minipage}{\textwidth}
\begin{center}
\begin{tikzpicture}
\def\h{.7}
\def\v{.5}
\def\nsep{.85}
\def\csep{.9}


\foreach \i in {0,...,2} {
	\pgfmathsetmacro{\x}{\i*\nsep*\h}
	\pgfmathsetmacro{\y}{0*\v}
	\pgfmathsetmacro{\lab}{\i+1}
	\node[draw,circle, minimum size = 22pt, inner sep = 1pt,scale=.6]
		(1p\i) at (\x,\y) {$\xi^{\pgfmathprintnumber{\lab}}_{0}$};

	\pgfmathsetmacro{\x}{\i*\nsep*\h}
	\pgfmathsetmacro{\y}{-2*\v}
	\pgfmathsetmacro{\lab}{\i+1}
	\node[draw,circle, minimum size = 22pt, inner sep = 1pt,scale=.6] 
		(2p\i) at (\x,\y) {$\xi^{\pgfmathprintnumber{\lab}}_{1}$};
}
\foreach \i in {3,...,5} {
	\pgfmathsetmacro{\x}{(\nsep*\i+\csep)*\h}
	\pgfmathsetmacro{\y}{0*\v}
	\pgfmathsetmacro{\lab}{\i+1}
	\node[draw,circle, minimum size = 22pt, inner sep = 1pt,scale=.6]
		(1p\i) at (\x,\y) {$\xi^{\pgfmathprintnumber{\lab}}_{0}$};

	\pgfmathsetmacro{\x}{(\nsep*\i+\csep)*\h}
	\pgfmathsetmacro{\y}{-2*\v}
	\pgfmathsetmacro{\lab}{\i+1}
	\node[draw,circle, minimum size = 22pt, inner sep = 1pt,scale=.6] 
		(2p\i) at (\x,\y) {$\xi^{\pgfmathprintnumber{\lab}}_{1}$};
}

\foreach \i in {0,...,5} {
	\foreach \j in {0,...,5} {
		\draw (1p\i) -- (2p\j);
	}
}

\foreach \i in {6,...,8} {
	\pgfmathsetmacro{\x}{(\nsep*\i+2*\csep)*\h}
	\pgfmathsetmacro{\y}{0*\v}
	\pgfmathsetmacro{\lab}{\i+1}
	\node[draw,circle, minimum size = 22pt, inner sep = 1pt,scale=.6]
		(1p\i) at (\x,\y) {$\xi^{\pgfmathprintnumber{\lab}}_{0}$};
	\pgfmathsetmacro{\x}{(\nsep*\i+2*\csep)*\h}
	\pgfmathsetmacro{\y}{-2*\v}
	\pgfmathsetmacro{\lab}{\i+1}
	\node[draw,circle, minimum size = 22pt, inner sep = 1pt,scale=.6] 
		(2p\i) at (\x,\y) {$\xi^{\pgfmathprintnumber{\lab}}_{1}$};
}
\foreach \i in {9,...,11} {
	\pgfmathsetmacro{\x}{(\nsep*\i+3*\csep)*\h}
	\pgfmathsetmacro{\y}{0*\v}
	\pgfmathsetmacro{\lab}{\i+1}
	\node[draw,circle, minimum size = 22pt, inner sep = 1pt,scale=.6]
		(1p\i) at (\x,\y) {$\xi^{\pgfmathprintnumber{\lab}}_{0}$};
	\pgfmathsetmacro{\x}{(\nsep*\i+3*\csep)*\h}
	\pgfmathsetmacro{\y}{-2*\v}
	\pgfmathsetmacro{\lab}{\i+1}
	\node[draw,circle, minimum size = 22pt, inner sep = 1pt,scale=.6] 
		(2p\i) at (\x,\y) {$\xi^{\pgfmathprintnumber{\lab}}_{1}$};
}

\foreach \i in {6,...,11} {
	\foreach \j in {6,...,11} {
		\draw (1p\i) -- (2p\j);
	}
}

\foreach \i in {0,...,2} {
	\pgfmathsetmacro{\x}{(\nsep*\i)*\h}
	\pgfmathsetmacro{\y}{-6*\v}
	\pgfmathsetmacro{\lab}{\i+1}
	\node[draw,circle, minimum size = 22pt, inner sep = 1pt,scale=.6] 
		(3p\i) at (\x,\y) {$\xi^{\pgfmathprintnumber{\lab}}_{2}$};
}
\foreach \i in {3,...,5} {
	\pgfmathsetmacro{\x}{(\nsep*\i+1*\csep)*\h}
	\pgfmathsetmacro{\y}{-6*\v}
	\pgfmathsetmacro{\lab}{\i+1}
	\node[draw,circle, minimum size = 22pt, inner sep = 1pt,scale=.6] 
		(3p\i) at (\x,\y) {$\xi^{\pgfmathprintnumber{\lab}}_{2}$};
}
\foreach \i in {6,...,8} {
	\pgfmathsetmacro{\x}{(\nsep*\i+2*\csep)*\h}
	\pgfmathsetmacro{\y}{-6*\v}
	\pgfmathsetmacro{\lab}{\i+1}
	\node[draw,circle, minimum size = 22pt, inner sep = 1pt,scale=.6] 
		(3p\i) at (\x,\y) {$\xi^{\pgfmathprintnumber{\lab}}_{2}$};
}
\foreach \i in {9,...,11} {
	\pgfmathsetmacro{\x}{(\nsep*\i+3*\csep)*\h}
	\pgfmathsetmacro{\y}{-6*\v}
	\pgfmathsetmacro{\lab}{\i+1}
	\node[draw,circle, minimum size = 22pt, inner sep = 1pt,scale=.6] 
		(3p\i) at (\x,\y) {$\xi^{\pgfmathprintnumber{\lab}}_{2}$};
}

\foreach \i in {0,...,2} {
	\foreach \j in {0,...,2} {
		\draw (3p\i) -- (2p\j);
	}
}
\foreach \i in {0,...,2} {
	\foreach \j in {6,...,8} {
		\draw (3p\i) -- (2p\j);
	}
}

\foreach \i in {3,...,5} {
	\foreach \j in {3,...,5} {
		\draw (3p\i) -- (2p\j);
	}
}
\foreach \i in {3,...,5} {
	\foreach \j in {9,...,11} {
		\draw (3p\i) -- (2p\j);
	}
}

\foreach \i in {6,...,8} {
	\foreach \j in {6,...,8} {
		\draw (3p\i) -- (2p\j);
	}
}
\foreach \i in {6,...,8} {
	\foreach \j in {0,...,2} {
		\draw (3p\i) -- (2p\j);
	}
}

\foreach \i in {9,...,11} {
	\foreach \j in {9,...,11} {
		\draw (3p\i) -- (2p\j);
	}
}
\foreach \i in {9,...,11} {
	\foreach \j in {3,...,5} {
		\draw (3p\i) -- (2p\j);
	}
}

\pgfmathsetmacro{\x}{(2*\nsep*\h + (\nsep*3+1*\csep)*\h)/2}
\draw[dashed] (\x,1*\v) -- (\x,-6.5*\v);
\pgfmathsetmacro{\x}{((8*\nsep+2*\csep)*\h + (\nsep*9+3*\csep)*\h)/2}
\draw[dashed] (\x,1*\v) -- (\x,-6.5*\v);
\pgfmathsetmacro{\x}{((5*\nsep+1*\csep)*\h + (\nsep*6+2*\csep)*\h)/2}
\draw[dashed] (\x,1*\v) -- (\x,-6.5*\v);

\pgfmathsetmacro{\x}{1*\nsep*\h}
\node[] at (\x,1*\v) {PE1};
\pgfmathsetmacro{\x}{(4*\nsep+1*\csep)*\h}
\node[] at (\x,1*\v) {PE2};
\pgfmathsetmacro{\x}{(7*\nsep+2*\csep)*\h}
\node[] at (\x,1*\v) {PE3};
\pgfmathsetmacro{\x}{(10*\nsep+3*\csep)*\h}
\node[] at (\x,1*\v) {PE4};


\end{tikzpicture}
\end{center}
\end{minipage} \\
 
\begin{minipage}{\textwidth}
\begin{equation*}
A_{1} = 
\frac{1}{2}
\begin{bmatrix}
\ones_{1/3} & \ones_{1/3} & 0 & 0 \\
\ones_{1/3} & \ones_{1/3} & 0 & 0 \\
0 & 0 & \ones_{1/3} & \ones_{1/3} \\
0 & 0 & \ones_{1/3} & \ones_{1/3} 
\end{bmatrix}, \quad 
A_{2} = 
\frac{1}{2}
\begin{bmatrix}
\ones_{1/3} & 0 & \ones_{1/3} & 0 \\
0 & \ones_{1/3} & 0 & \ones_{1/3} \\
\ones_{1/3} & 0 & \ones_{1/3} & 0\\
0 & \ones_{1/3} & 0 & \ones_{1/3}
\end{bmatrix}
\end{equation*}
\end{minipage}
\end{tabular}
\end{center}
\caption{Augmented resampling with $m=4$ and $M=3$. Dashed vertical lines separate the groups of particles belonging to different PEs.}
\label{fig:BR2PF}
\end{figure}


Figure \ref{fig:BR2PF} illustrates the matrices $A_{1},\ldots,A_{S}$ and how they determine the pairs of interacting PEs at different stages. Each node in the graph represents an individual particle at a specific stage. An edge between $\xi^{i}_{s-1}$ and $\xi^{j}_{s}$ is equivalent to $A^{ij}_{s} \neq 0$ and hence the particle $\xi^{i}_{s}$ is sampled with replacement among the particles $\xi^{j}_{s-1}$ where $j$ is such that $A^{ij}_{s} \neq 0$. At stage $s$ the PE containing particle $\xi^{i}_{s}$ is thus required to communicate only with the PE containing the elements of $(\xi^{1}_{s-1},\ldots,\xi^{N}_{s-1})$ connected to $\xi^{i}_{s}$ and, as demonstrated in Figure \ref{fig:BR2PF}, only pairwise interactions between PEs are required; at the first stage the interacting pairs are (PE1,PE2) and (PE3,PE4) and at the second stage (PE1,PE3) and (PE2,PE4). This \emph{radix-2 butterfly} structure (see e.g.~\cite{oppenheim75}) of Figure \ref{fig:BR2PF} will be formally stated in Section \ref{sec:augmented analysis}.

There are two motivations for our interest in studying augmented resampling in the particle filtering context. The first is related to the communication pattern between PEs and the second is related to adaptive resampling schemes. 

Regarding the communication pattern, let us assume an idealised computer architecture in which a PE can communicate with at most one other PE at a time and different pairs of PEs can communicate perfectly in parallel. Moreover, we assume that the time required to perform the communication is constant over the pairs of PEs. We acknowledge that in reality these assumptions are only approximate as computer architectures involve various types of PEs (e.g.~networks of computers or cores within processors) interconnected by various network topologies (e.g.~hypercubes or data buses). 

Now suppose that there are four PEs (PE1, PE2, PE3, and PE4) and only the sample contained in a single PE, say PE1, has an effectively non-zero weight. In this case, without augmented resampling, PE1 would have to disseminate its sample to all other PEs; first to PE2, then to PE3 and finally to PE4. This suggests that $m-1$ sequential communication steps are needed. With augmented resampling, as in Figure \ref{fig:BR2PF}, PE1 would first send its sample to PE2, after which PE1 would send the sample to PE3 while \emph{at the same} time PE2 could send its sample (just received from PE1) to PE4 thereby accomplishing complete dissemination of PE1's sample in only $\log_{2}(m)$ sequential communication steps, making augmented resampling apparently more efficient in terms of communication. However, to account for the fact that our reasoning here is based on the idealised model, we will base our final conclusions on numerical experiments.

The second motivation for augmented resampling is its additional flexibility in adaptive resampling schemes \citep{liu_et_chen98}. It is well known that although resampling is the enabling factor for long term stability of SMC methods, it does introduce error as well as additional computational cost and should only be done when necessary. In adaptive resampling, prior to performing the actual resampling, one first evaluates the \emph{effective sample size} (ESS)~\citep{liu_et_chen95} which for the BPF can be formally expressed as
\begin{equation*}
\mathcal{E}_{n} = \dfrac{\left(N^{-1}\sum_{i=1}^{N} g_{n}(\zeta^{i}_{n})\right)^2}{N^{-1}\sum_{i=1}^{N} g^{2}_{n}(\zeta^{i}_{n})} \in \left[\frac{1}{N},1\right],
\end{equation*}
and executes the resampling step only if $\mathcal{E}_{n}$ is below some predetermined threshold $\theta \in (1/N,1]$. This means that every filter iteration with adaptive resampling involves a dichotomous decision to either allow the full interaction of all particles or allow no interaction at all. Augmented resampling enables this decision to be refined so that the decision is made between finer levels of interaction, and hence it may be possible to find a better balance between long term stability, resampling error, and computational cost. In practice this is accomplished by evaluating the ESS after every stage of augmented resampling and, based on the ESS, deciding whether to proceed to the next resampling stage or to skip the remaining resampling stages and move on to the next time step. 

This more flexible adaptation of resampling is based on the ideas presented in \citep{whiteley_et_al14} and it will lead to an increase in efficiency if sufficiently few resampling stages in total are executed. It is also worth noting that the evaluation of the ESS at every stage introduces some additional communication, but our numerical experiments suggest that the net effect of this \emph{fully adapted} resampling scheme is a notable gain in efficiency. The rigorous theoretical analysis of the convergence properties of this method is left beyond the scope of this paper.

This paper is organised as follows. Section \ref{sec:augmented analysis} is dedicated to the theoretical properties of augmented resampling outside the particle filtering context and it presents our main convergence result for augmented resampling, namely Proposition \ref{prop:intro_to_aug_resampling}. In Section \ref{sec:r2pf} we apply the augmented resampling algorithm in the particle filter context and present our main convergence result, Theorem \ref{thm:convergence}. Section \ref{sec:another r2 resampler} introduces a modified augmented resampling scheme reminiscent to that used in IPF and the convergence of the resulting particle filter is proved in Section \ref{sec:grouped r2pf}. Section \ref{sec:numerics} concludes the paper with some results of numerical experiments showing the potential of the proposed algorithms and a brief discussion on the conclusions. Most of the more technical proofs are housed in the appendices.

\subsection{Notations}

We let $\boundMeas(\X)$ denote the bounded and measurable $\R$ valued functions defined on $(\X,\cX)$. Throughout the paper, we define $\infnorm{\varphi} \defeq \sup_{x \in \X}|\varphi(x)|$ and $\osc{\varphi} \defeq \sup_{x,y\in \X}|\varphi(x)-\varphi(y)|$ for any $\varphi \in \boundMeas(\X)$. We define two specific subsets of $\boundMeas(\X)$, $\boundMeas_{+}(\X) \defeq \{\varphi \in \boundMeas(\X): \varphi > 0\}$ and $\boundMeas_{1}(\X) \defeq \{\varphi \in \boundMeas(\X): \infnorm{\varphi} \leq 1\}$.  For a sequence of square matrices $(A_{s})_{1\leq s \leq S}$ where $S \in \bN_{+}$ we write $\prod_{s=1}^{S} A_{s} = A_{S}\cdots A_{1}$. For any $N,S\in \{1,2,\ldots\}$ we use the shorthand notation $\sum_{(i_0,\ldots,i_S)} \defeq \sum_{i_0=1}^{N}\cdots\sum_{i_S=1}^{N}$. We also define $\lceil x\rceil \defeq \min\{z \in \mathbb{Z}: z \geq x\}$ and $\lfloor x\rfloor \defeq \max\{z \in \mathbb{Z}: z \leq x\}$ and $( x \bmod z ) \defeq x - z\lfloor x/z\rfloor$. Throughout the remainder of this paper $\E$ and $\P$ refer to the expectation and probability with respect to the probability space charactering the randomness of the algorithm only. The observations of the underlying HMM are assumed fixed.

%
\section{Augmented resampling}
\label{sec:augmented analysis}

We start with a study of Algorithm \ref{alg:augmented_resampling} outside the filtering context by applying it to an arbitrary $\X^{N}$ valued random sample $\vxi_{0} = (\xi^{1}_{0},\ldots,\xi^{N}_{0})$ and an arbitrary weighting function $g\in \boundMeas_{+}(\X)$. We have the following result:
\begin{proposition}\label{prop:intro_to_aug_resampling}
Assume \eqref{ass:NS} and let $g\in\boundMeas_{+}(\ss)$. Then for any $\X^{N}$ valued random variable $\vxi_{0}$ and for any $\varphi\in\boundMeas(\ss)$,
\begin{equation}
\E\left[\left.\frac{1}{N}\sum_{i=1}^{N} \varphi(\xi_{S}^i)\right|\vxi_{0}\right]=\dfrac{\sum_{i=1}^{N} g(\xi_{0}^i)\varphi(\xi_{0}^i)}{\sum_{i=1}^{N} g(\xi_{0}^i)},\label{eq:aug_res_lack_of_bias}
\end{equation}
and for any $r\geq1$ there exists a finite constant $B_{r}$, depending only on $r$, such that no matter what the distribution of $\vxi_{0}$ is, we have
\begin{equation}
\E\Bigg[\Bigg|\Bigg(\frac{1}{N}\sum_{i=1}^{N} g(\xi_{0}^i)\Bigg)\Bigg(\frac{1}{N}\sum_{i=1}^{N} \varphi(\xi_{S}^i)\Bigg)-\frac{1}{N}\sum_{i=1}^{N} g(\xi_{0}^i)\varphi(\xi_{0}^i)\Bigg|^{r}\Bigg]^{\frac{1}{r}}
\leq B_{r}\sqrt{\frac{S}{N}} \infnorm{g}\osc{\varphi}. \label{eq:aug_res_L_p_bound}
\end{equation}
\end{proposition}

An immediate consequence of Proposition \ref{prop:intro_to_aug_resampling} is that if, for example, $S$ is some non-decreasing function $S(N)$ of $N$ such that $\sum_{N=1}^\infty(S(N)/N)^{r/2}<\infty$ for some $r\geq1$, then
\begin{equation}\label{eq:martingale error}
\Bigg(\frac{1}{N}\sum_{i=1}^{N} g(\xi_{0}^i)\Bigg)\Bigg(\frac{1}{N}\sum_{i=1}^{N} \varphi(\xi_{S}^i)-\frac{\sum_{i=1}^{N} g(\xi_{0}^i)\varphi(\xi_{0}^i)}{\sum_{i=1}^{N} g(\xi_{0}^i)}\Bigg) \;\almostsurelyOneArg{N\rightarrow\infty}\; 0,
\end{equation}
without requiring any convergence of $N^{-1} \sum_{i=1}^{N} g(\xi_{0}^i)$ or $N^{-1} \sum_{i=1}^{N} g(\xi_{0}^i)\varphi(\xi_{0}^i)$.

The more technical proofs of the results stated in this section are housed in Appendix \ref{sec:proofs augmented}.

\subsection{Properties of augmented resampling}

The matrices $A_{1},\ldots,A_{S}$ play an important role in augmented resampling and to a large extent they determine its statistical properties. We present first the following result which, although not in its entirety required to prove Proposition \ref{prop:intro_to_aug_resampling}, summarises some key properties of $A_{1},\ldots,A_{S}$ and also makes it formally explicit, how the structure of the diagram in Figure \ref{fig:BR2PF} is obtained.

\begin{lemma}\label{lem:facts about A}
Assume \eqref{ass:NS}. Then for all $1 \leq s \leq S$, $A_{s}$ is symmetric, idempotent, and doubly stochastic. Moreover, for any $1 \leq i \leq m$
\begin{align}\label{eq:explicit nonzero elements}
&\left\{j \in \{1,\ldots,m\}: (\Id_{2^{S-s}} \otimes \ones_{1/2} \otimes \Id_{2^{s-1}})^{ij} \neq 0 \right\} \nonumber\\
&\quad= \left\{ \left((i-1)\bmod 2^{s-1}\right)+(q-1)2^{s-1}+2^{s}\left\lfloor\frac{i-1}{2^{s}}\right\rfloor+1:q \in \{1,2\}\right\},  
\end{align}
and for all $1 \leq i \leq m$, $(\Id_{2^{S-s}} \otimes \ones_{1/2} \otimes \Id_{2^{s-1}})^{ii} = 1/2$.
\end{lemma}
Equation \eqref{eq:explicit nonzero elements} formalises the radix-2 butterfly structure seen in Figure \ref{fig:BR2PF} by giving explicit expression for the nonzero elements of $\Id_{2^{S-s}} \otimes \ones_{1/2} \otimes \Id_{2^{s-1}} \in \R^{m\times m}$. By considering $A_{s}\in \R^{mM\times mM}$ as an $m$-by-$m$ matrix of $M$-by-$M$ blocks, the element $(i,j)$ of $\Id_{2^{S-s}} \otimes \ones_{1/2} \otimes \Id_{2^{s-1}}$ is nonzero if and only if the block $(i,j)$ of $A_{s}$ is the full matrix $\frac{1}{2}\ones_{1/M}$.

From Algorithm \ref{alg:augmented_resampling} we obtain the definitions 
\begin{equation}
V^{i}_{0}\defeq g(\xi^{i}_{0}),\qquad V^{i}_{s}\defeq \sum_{j=1}^{N} A^{ij}_{s}V^{j}_{s-1},\qquad 1\leq i\leq N,~ 1 \leq s \leq S\label{eq:V_defn_proofs}
\end{equation}
for the particle weights at each stage. For the proof of Proposition \ref{prop:intro_to_aug_resampling} it is crucial that after finishing all $S$ resampling stages, Algorithm \ref{alg:augmented_resampling} returns an \emph{unweighted} sample in a manner similar to conventional multinomial resampling, i.e.~that $V_{S}^{i} = V_{S}^{j}$ for all $i,j \in \{1,\ldots,N\}$. The proof of this unweighted property is essentially due to the following key result on 
$A_{1},\ldots,A_{S}$.
\begin{lemma}\label{lem:ones product}
Assume \eqref{ass:NS}. Then $\prod_{s=1}^{S} A_{s} = \ones_{1/N}$.
\end{lemma}
Lemma \ref{lem:ones product} enables us to establish the following result which, in addition to the unweighted property, states some other facts about the weights $V^{i}_{s}$ that are required for the proof of Proposition \ref{prop:intro_to_aug_resampling}.\begin{lemma}\label{lem:facts_about_Vs}
Assume \eqref{ass:NS}. For any $1\leq i\leq N$ and $0 \leq s \leq S$,
\begin{enumerate}
\item \label{it:measurability of V} $V^{i}_{s}$ is measurable w.r.t.~$\sigma(\xi_{0})$,
\item \label{it:boundedness of V} $V^{i}_{s}\leq\infnorm{g}$.
\item \label{it:constant weights} $V^{i}_{S}=N^{-1}\sum_{j=1}^{N} g(\xi_{0}^j)$.
\end{enumerate}
\end{lemma}
\begin{remark}
Although we work throughout the paper with the definition \eqref{eq:A matrices} of $A_{s}$, the specific definition of $(A_{s})_{1\leq s \leq S}$ is irrelevant for the proof of Proposition \ref{prop:intro_to_aug_resampling} as long as the matrices satisfy Lemma \ref{lem:ones product} and are doubly stochastic. Different definitions of $(A_{s})_{1\leq s \leq S}$ for which Lemma \ref{lem:ones product} still holds can be easily devised. Above, the double stochasticity follows from Lemma \ref{lem:facts about A}.
\end{remark}

\subsection{Proof of Proposition \ref{prop:intro_to_aug_resampling} by martingale difference}

The proof of Proposition \ref{prop:intro_to_aug_resampling} is based on expressing the error term on the left hand side of \eqref{eq:martingale error} as a martingale to which we then apply the Burkholder inequality. For the required martingale construction, we observe another important property of Algorithm \ref{alg:augmented_resampling}; for all $1\leq s \leq S$ the random samples $\vxi_{s} \defeq (\xi^{1}_{s},\ldots,\xi^{N}_{s})$ satisfy a \emph{one step conditional independence} property, i.e.~the particles $\xi^{1}_{s},\ldots,\xi^{N}_{s}$ are conditionally independent given $\vxi_0,\ldots,\vxi_{s-1}$. We also see that for each $1\leq i\leq N$ and $B\in\cX$, we have
\begin{equation}
\P\big(\xi_s^i\in B \,\big|\, \vxi_0,\ldots,\vxi_{k-1}\big)=\frac{1}{V_s^i}\sum_{j=1}^{N} A_s^{ij}V_{s-1}^j \mathbb{I}_B(\xi_{s-1}^j),\label{eq:P_in_terms_of_V_proofs}
\end{equation}
where $\mathbb{I}_B$ denotes the indicator function of the set $B\in \cX$.

To construct the required martingale via a martingale difference, we define a sequence $\cM \defeq \{(X_{\rho},\cF_{\rho});\, 0 \leq \rho \leq SN\}$ where $X_{0} \defeq 0$, $\cF_{0} \defeq \sigma(\vxi_{0})$ and for all $0 < \rho \leq SN$, we define 
\begin{align*}
X_{\rho} &\defeq
\frac{V^{i_{N}(\rho)}_{s_{N}(\rho)}}{\sqrt{SN}}\left(\barphi(\xi^{i_{N}(\rho)}_{s_{N}(\rho)}) - \frac{1}{V^{i_{N}(\rho)}_{s_{N}(\rho)}} \sum_{j=1}^{N} A^{i_{N}(\rho)j}_{s_{N}(\rho)}V^{j}_{s_{N}(\rho)-1}\barphi(\xi^{j}_{s_{N}(\rho)-1})\right), \\
\cF_{\rho} &\defeq \cF_{\rho-1} \vee \sigma(\xi^{i_{N}(\rho)}_{s_{N}(\rho)}) 
\end{align*}
where
\begin{equation*}
\barphi \defeq \varphi - \frac{\sum_{i=1}^{N}g(\xi^{i}_{0})\varphi(\xi^{i}_{0})}{\sum_{i=1}^{N}g(\xi^{i}_{0})}, \qquad \varphi \in \boundMeas(\X),
\end{equation*}
and for any $k\in \mathbb{N}$ 
\begin{equation}\label{eq:index map}
i_{k}(\rho) \defeq ((\rho - 1) \bmod k)+1  \qquad s_k(\rho) \defeq  \left\lceil\frac{\rho}{k}\right\rceil.
\end{equation}
The purpose of \eqref{eq:index map} is simply to define a bijective index map taking a one dimensional index $\rho$ in the range $\{1,\ldots,Sk\}$ into a pair of indices $s_{k}(\rho)\in \{1,\ldots,S\}$ and $i_{k}(\rho) \in \{1,\ldots,k\}$. The following proposition establishes the required martingale properties of $\cM$.
\begin{proposition}\label{prop:martingale decomposition}
Assume \eqref{ass:NS}. The following statements hold:
\begin{enumerate}

\item\label{it:measurability} $X_{\rho}$ is $\cF_{\rho}$-measurable for all $0 \leq \rho \leq SN$;

\item\label{it:zero mean} $\E\left[X_{\rho}\,\mid\,\cF_{\rho-1}\right] = 0$ (a.s.) for all $0 < \rho \leq SN$;

\item\label{it:boundedness} $\left|X_{\rho}\right| \leq \infnorm{g}\osc{\varphi}/\sqrt{SN}$ for all $0 \leq \rho \leq SN$;

\item and we have the identities
\begin{align}
\sqrt{\frac{S}{N}}\sum_{\rho=1}^{SN} X_{\rho} &= \frac{1}{N}\sum_{i=1}^{N} V^{i}_{S}\barphi(\xi^{i}_{S}) \label{eq:first decompo}\\
&= \left(\frac{1}{N}\sum_{i=1}^{N}g(\xi^{i}_{0})\right)\left(\frac{1}{N}\sum_{i=1}^{N}\varphi(\xi^{i}_{S})\right) - \frac{1}{N}\sum_{i=1}^{N} g(\xi^{i}_{0})\varphi(\xi^{i}_{0}) \label{eq:second decompo}
\end{align}
\end{enumerate}
\end{proposition}
\if\proofs1
\begin{proof}
By the definitions of $X_{\rho}$ and $\cF_{\rho}$ we have \eqref{it:measurability} by Lemma \ref{lem:facts_about_Vs}\eqref{it:measurability of V}. Claim \eqref{it:zero mean} follows from the one step conditional independence and \eqref{eq:P_in_terms_of_V_proofs}. Claim \eqref{it:boundedness} follows from Lemma \ref{lem:facts_about_Vs}\eqref{it:boundedness of V}. Since $N^{-1}\sum_{i=1}^{N}V^{i}_{0} \barphi(\xi^{i}_{0}) = 0$, we have the decomposition
\begin{align*}
\frac{1}{N}\sum_{i=1}^{N} V^{i}_{S}\barphi(\xi^{i}_{S}) &= 
\sum_{q = 1}^{S} \left(
\frac{1}{N}\sum_{i_{q}=1}^{N} V^{i_{q}}_{q}\barphi(\xi^{i_{q}}_{q}) 
- \frac{1}{N}\sum_{i_{q-1}=1}^{N} V^{i_{q-1}}_{q-1}\barphi(\xi^{i_{q-1}}_{q-1})
\right)
\end{align*}
Because $A_{q}$ is double stochastic we have $\sum_{j=1}^{N} A_{q}^{ji}=1$ and hence
\begin{align*}
&\frac{1}{N}\sum_{i_{q}=1}^{N} V^{i_{q}}_{q}\barphi(\xi^{i_{q}}_{q}) 
- \frac{1}{N}\sum_{i_{q-1}=1}^{N}V^{i_{q-1}}_{q-1}\barphi(\xi^{i_{q-1}}_{q-1}) \\
&= \frac{1}{N}\sum_{i_{q}=1}^{N} V^{i_{q}}_{q}\barphi(\xi^{i_{q}}_{q}) 
- \frac{1}{N}\sum_{j=1}^{N}\sum_{i_{q-1}=1}^{N}A^{ji_{q-1}}_{q}V^{i_{q-1}}_{q-1}\barphi(\xi^{i_{q-1}}_{q-1}) \\
&= \frac{1}{N}\sum_{i_{q}=1}^{N} V^{i_{q}}_{q}\left(\barphi(\xi^{i_{q}}_{q}) - \frac{1}{V^{i_{q}}_{q}}\sum_{i_{q-1}=1}^{N}A^{i_{q}i_{q-1}}_{q}V^{i_{q-1}}_{q-1}\barphi(\xi^{i_{q-1}}_{q-1})\right) \\
&= \sqrt{\frac{S}{N}}\sum_{i=1}^{N} X_{(q-1)N+i},
\end{align*}
proving \eqref{eq:first decompo}.

By Lemma \ref{lem:facts_about_Vs}\eqref{it:constant weights} and \eqref{eq:V_defn_proofs}, we can prove \eqref{eq:second decompo} by writing
\begin{align*}
&\left(\frac{1}{N}\sum_{i=1}^{N}g(\xi^{i}_{0})\right)\left(\frac{1}{N}\sum_{i=1}^{N}\varphi(\xi^{i}_{S})\right) - \frac{1}{N}\sum_{i=1}^{N} g(\xi^{i}_{0})\varphi(\xi^{i}_{0})\\
&=\frac{1}{N}\sum_{i=1}^{N}V^{i}_{S}\varphi(\xi^{i}_{S})
- \frac{1}{N}\sum_{i=1}^{N} g(\xi^{i}_{0})\varphi(\xi^{i}_{0}) \\
&= \frac{1}{N}\sum_{i=1}^{N} V^{i}_{S}\left(\varphi(\xi^{i}_{S}) - \frac{\sum_{i=1}^{N} g(\xi^{i}_{0})\varphi(\xi^{i}_{0})}{\sum_{i}g(\xi^{i}_{0})}\right) \\
&= \frac{1}{N}\sum_{i=1}^{N} V^{i}_{S}\barphi(\xi^{i}_{S})\qedhere
\end{align*}
\end{proof}
\fi

By Proposition \ref{prop:martingale decomposition} the proof of Proposition \ref{prop:intro_to_aug_resampling} is obtained readily as follows.
\begin{proof}[Proof of Proposition \ref{prop:intro_to_aug_resampling}]
The lack of bias \eqref{eq:aug_res_lack_of_bias} follows by Proposition \ref{prop:martingale decomposition}(\ref{it:zero mean}), \eqref{eq:first decompo}, \eqref{eq:second decompo}, and the tower property of conditional expectations. Bound \eqref{eq:aug_res_L_p_bound} follows by Burkholder-Davis-Gundy inequality and Proposition \ref{prop:martingale decomposition}(\ref{it:boundedness}) by writing
\begin{equation*}
\E\left[\left|\sum_{\rho=1}^{SN} X_{\rho}\right|^r\right] \leq B^{r}_{r}\E\left[\left|\sqrt{\sum_{\rho=1}^{SN} \left|X_{\rho}\right|^2}\right|^{r}\right] \leq B^{r}_{r} \infnorm{g}^{r} \osc{\varphi}^{r}.
\end{equation*}

\end{proof}

\section{Particle filter with augmented resampling}
\label{sec:r2pf}

We now turn to analysing the implications of replacing Algorithm \ref{alg:multinomial_resampling} in BPF with Algorithm \ref{alg:augmented_resampling}. The following mild regularity condition  on the underlying HMM is assumed to hold.
\begin{assumption}\label{ass:potential}
For all $n\geq 0$, $g_{n} \in \boundMeas_{+}(\X)$.
\end{assumption}
\noindent
Under \eqref{ass:potential}, we show that the resulting particle filter is convergent in mean (of order $r\geq 1$). In order to establish uniform in time convergence in mean, the following strong but standard regularity assumption is made \citep{whiteley_et_al14, delmoral04}.
\begin{assumption}\label{ass:regularity}
There exists $\delta \geq 1$ and $\epsilon \in (0,1)$ such that
\begin{equation*}
\sup_{n\geq 0}\sup_{x,y} \frac{g_{n}(x)}{g_{n}(y)} \leq \delta, \qquad \text{and} \qquad f(x,\,\cdot\,) \geq \epsilon f(y,\,\cdot\,),~\forall x,y \in \X^2.
\end{equation*}
\end{assumption}
\noindent

\begin{theorem}\label{thm:convergence} 
Fix $N$ and $S$ and assume \eqref{ass:NS} and \eqref{ass:potential}. If the measures $(\pi^{N}_{n})_{n\geq 0}$ are calculated by Algorithm \ref{alg:bpf} deploying Algorithm \ref{alg:augmented_resampling}, then we have the following:
\begin{enumerate}
\item \label{it:convergence} For all $n\geq 0$ and $r\geq 1$, there exists $C_{n,r} \in \R_{+}$ such that
\begin{equation}\label{eq:convergence ind ass}
\sup_{\varphi\in\boundMeas_{1}(\X)} \E\left[\left|\pi^{N}_{n}(\varphi) - \pi_{n}(\varphi)\right|^{r}\right]^{\frac{1}{r}} \leq C_{n,r}\sqrt{\frac{S}{N}}.
\end{equation}

\item \label{it:uniform convergence} If in addition \eqref{ass:regularity} holds, then for all $r\geq 1$ there exists $C_{r} \in \R_{+}$ such that
\begin{equation*}
\sup_{n\geq 0}\sup_{\varphi\in\boundMeas_{1}(\X)} \E\left[\left|\pi^{N}_{n}(\varphi) - \pi_{n}(\varphi)\right|^{r}\right]^{\frac{1}{r}} \leq C_{r}\sqrt{\frac{S}{N}}.
\end{equation*}
\end{enumerate}
\end{theorem}
\noindent
Although Theorem \ref{thm:convergence} resembles many existing results on BPF, and its variations, the interpretation is somewhat different. The result is stated under the assumption \eqref{ass:NS} which leaves the convergence rate ambiguous. However, if we write the r.h.s.~of \eqref{eq:convergence ind ass} in terms of $m$ and $M$, we observe that by fixing $m$, \eqref{eq:convergence ind ass} yields the standard $M^{-1/2}$ rate of convergence, and by fixing $M$, a slower $\sqrt{\log_{2}(m)/m}$ rate is obtained. The rate is slower due to the numerator term $\sqrt{S} = \sqrt{\log_{2}(m)}$ which can be intuitively interpreted to trace back to the resampling errors introduced at each stage of augmented resampling. In both cases, by Borel-Cantelli argument, Theorem \ref{thm:convergence} also yields the law of large numbers, i.e.~that $\pi^{mM}_{n}(\varphi) - \pi_{n}(\varphi) \to 0$ almost surely as $m\to\infty$ (resp. $M\to\infty$) and $M$ (resp. $m$) is kept fixed.


While the convergence rate $M^{-1/2}$ that we obtain for fixed $m$ is known to be optimal, the analysis that we carry out to prove Theorem \ref{thm:convergence} does not explicitly imply that also the $\sqrt{\log_{2}(m)/m}$ rate, obtained for fixed $M$, is optimal. However, we conjecture this begin the case. Some evidence supporting this conjecture is given in the unpublished work \citep{heine_et_al_14}, where a CLT for a similar, but not identical, algorithm was shown to have the same scaling factor.



In the following subsections we go through the steps of proving Theorem \ref{thm:convergence}. The more technical proofs are postponed to Appendix \ref{sec:proofs R2PF}.

\subsection{Preliminary results}

The proof of Theorem \ref{thm:convergence}(\ref{it:convergence}) is by induction. The following lemma, whose primary purpose is to initialise the induction, is a special instance of the more general result proved in \citep[Lemma 7.3.3]{delmoral04} and hence we omit the proof.
\begin{lemma}\label{lem:delmoral}
Let $(\zeta^{1},\ldots,\zeta^{N})$ be an i.i.d.~sample from some distribution $\pi$ defined on $(\X,\cX)$. Then there exists a constant $C^{\ast}_{r}\in \R_{+}$ depending only on $r$ such that
\begin{equation*}
\sup_{\varphi\in\boundMeas_{1}(\X)}\lpnorm{r}{\frac{1}{N}\sum_{i=1}^{N} \varphi(\zeta^{i})-\pi(\varphi)} \leq \frac{C^{\ast}_{r}}{\sqrt{N}}.
\end{equation*}
\end{lemma}
We also frequently use the following result to bound the error introduced by the mutation step of the particle filter. 
\begin{lemma}\label{lem:mutation martingale}
Let $(\widehat{\zeta}^{i},\ldots,\widehat{\zeta}^{N})$ be a $\X^{N}$ valued random variable and let
\begin{equation}\label{eq:joint law}
(\zeta^{i},\ldots,\zeta^{N}) \,\mid\, (\widehat{\zeta}^{i},\ldots,\widehat{\zeta}^{N})  \sim \prod_{i=1}^{N}f(\widehat{\zeta}^{i},\,\cdot\,).
\end{equation}
Then there exists a constant $B_{r} \in \R_{+}$ such that for all $N>0$
\begin{align*}
\sup_{\varphi\in\boundMeas_{1}(\X)}\lpnorm{r}{\frac{1}{N}\sum_{i=1}^{N}\varphi(\zeta^{i}) - \frac{1}{N}\sum_{i=1}^{N}f(\varphi)(\widehat{\zeta}^{i})} \leq \frac{2B_{r}}{\sqrt{N}}.
\end{align*}
\end{lemma}

Instead of a proof by induction, the proof of Theorem \ref{thm:convergence}(\ref{it:uniform convergence}) is based on the proof of Theorem 7.4.4 in~\citep{delmoral04}. For any probability measure $\mu$ on $(\X,\cX)$ and any $\varphi\in\boundMeas(\X)$, we define
\begin{equation*}
\Phi_{0}(\pi_{-1}^{N})\defeq\pi_{0}, \qquad\text{and}\qquad \Phi_{n}(\mu)(\varphi) \defeq \frac{\mu(g_{n-1}f(\varphi))}{\mu(g_{n-1})}, \qquad n>0.
\end{equation*}
We note that $\Phi_{n}$ is the mapping which generates the sequence of exact measures $(\pi_{n})_{n\geq0}$ by the recursion 
\begin{equation*}
\pi_{n} = \Phi_{n}(\pi_{n-1}), \qquad n \geq 0.
\end{equation*}
By using these notations, we have the following corollary of the proof of Theorem 7.4.4 in~\citep{delmoral04}. 
\begin{lemma}\label{lem:telescope reformulation}
Assume \eqref{ass:regularity}. Then for all $0 \leq n$, $0 \leq p \leq n$ and $\varphi \in \boundMeas_{1}(X)$, there exists $\alpha_{p,n} \in \R_{+}$ and $\varphi_{p,n,\varphi} \in \boundMeas_{1}(\X)$ such that 
\begin{equation*}
\big|\pi^{N}_{n}(\varphi) - \pi_{n}(\varphi)\big| \leq \sum_{p=0}^{n} \alpha_{p,n}\big|\big(\pi^{N}_{p}-\Phi_{p}(\pi^{N}_{p-1})\big)(\varphi_{p,n,\varphi})\big|
\end{equation*}
and $\sum_{p=0}^{n} \alpha_{p,n} \leq \delta/\epsilon^3$.
\end{lemma}

\subsection{Convergence}

Before the proof of Theorem \ref{thm:convergence}, we introduce an intermediate result, Proposition \ref{prop:induction step} below, consisting of two parts. The first part establishes the induction step needed for the proof of Theorem \ref{thm:convergence}\eqref{it:convergence}. The second part is used in the proof of Theorem \ref{thm:convergence}\eqref{it:uniform convergence} and it establishes a uniform bound for the local error terms $\pi^{N}_{n} - \Phi_{n}(\pi^{N}_{n-1})$ appearing in Lemma \ref{lem:telescope reformulation}.
%
For the brevity of notation we introduce the following probability measures
\begin{align}\label{eq:pi tilde and pi hat}
\widetilde{\pi}^{N}_{n} \defeq \frac{1}{N}\sum_{i=1}^{N} f(\widehat{\zeta}_{n-1}^{i},\,\cdot\,) \qquad 
\widehat{\pi}_{n}^{N} \defeq \frac{1}{N}\sum_{i=1}^{N}\delta_{\widehat{\zeta}^{i}_{n}}, \qquad n > 0. 
\end{align}
We also define
\begin{equation*}
\widehat{\pi}_{n}(\ud x) \defeq \frac{g_{n}(x)\pi_{n}(\ud x)}{\pi_{n}(g_{n})}, \qquad n \geq 0,
\end{equation*}
which is the exact filtering distribution associated with the HMM \eqref{eq:HMM}.

%
\begin{proposition}\label{prop:induction step}
Assume \eqref{ass:NS} and \eqref{ass:potential}. 
\begin{enumerate}
\item \label{it:step part a}
If for some $n\geq 0$ and some $r\geq 1$ there exists $C_{n,r} \in \R_{+}$ such that
\begin{equation}\label{eq:convergence precondition}
\sup_{\varphi\in\boundMeas_{1}(\X)}\E\left[\left| \pi^{N}_{n}(\varphi) - \pi_{n}(\varphi) \right|^{r}\right]^{\frac{1}{r}} \leq C_{n,r}\sqrt{\frac{S}{N}},
\end{equation} 
then there also exists $\widehat{C}_{n,r}\in \R$ such that
\begin{equation*}
\sup_{\varphi\in\boundMeas_{1}(\X)} \E\left[\left|\widehat{\pi}^{N}_{n}(\varphi) - \widehat{\pi}_{n}(\varphi)\right|^{r}\right]^{\frac{1}{r}} \leq \widehat{C}_{n,r}\sqrt{\frac{S}{N}}.
\end{equation*}
\item\label{it:step part b} If in addition \eqref{ass:regularity} holds, then for all $r\geq 1$ there exists $\widehat{C}_{r}\in \R_{+}$ such that
\begin{equation*}
\sup_{n\geq 0}\sup_{\varphi\in\boundMeas_{1}(\X)} \E\left[\left|\pi^{N}_{n}(\varphi) - \Phi_{n}(\pi^{N}_{n-1})(\varphi)\right|^{r}\right]^{\frac{1}{r}} \leq \widehat{C}_{r}\sqrt{\frac{S}{N}} .
\end{equation*}
\end{enumerate}
\end{proposition}
\noindent
Part \ref{it:step part a}) introduces the precondition \eqref{eq:convergence precondition} to bound the local error which effectively leads to the proof of Theorem \ref{thm:convergence}\eqref{it:convergence} being by induction. Under the assumption \eqref{ass:regularity} in part \ref{it:step part b}) such condition is not needed and the analysis becomes somewhat simpler.

\begin{proof}[Proof of Theorem \ref{thm:convergence}]
Fix $r\geq 1$. The proof of part \ref{it:convergence}) is by induction in $n \geq 0$.  
The induction is initialised by observing that at rank $n=0$, \eqref{eq:convergence ind ass} holds by Lemma \ref{lem:delmoral}. Suppose now that \eqref{eq:convergence ind ass} holds at some rank $n\geq 0$. By Minkowski's inequality, and the fact that $\pi_{n+1}(\varphi) = \widehat{\pi}_{n}(f(\varphi))$, we have
\begin{align*}
\sup_{\varphi\in\boundMeas_{1}(\X)}\lpnorm{r}{\pi^{N}_{n+1}(\varphi) - \pi_{n+1}(\varphi)} &\leq \sup_{\varphi\in\boundMeas_{1}(\X)}\lpnorm{r}{\pi^{N}_{n+1}(\varphi) - \widehat{\pi}^{N}_{n}(f(\varphi))} \\
&~+\sup_{\varphi\in\boundMeas_{1}(\X)}\lpnorm{r}{\widehat{\pi}^{N}_{n}(f(\varphi)) - \widehat{\pi}_{n}(f(\varphi))}.
\end{align*}
By applying Lemma \ref{lem:mutation martingale} and Proposition \ref{prop:induction step}, respectively, to the first and the second term on the r.h.s., we obtain the bound
\begin{align*}
\sup_{\varphi\in\boundMeas_{1}(\X)}\lpnorm{r}{\pi^{N}_{n+1}(\varphi) - \pi_{n+1}(\varphi)} &\leq 2B_{r}\sqrt{\frac{S}{N}} + \widehat{C}_{n,r}\sqrt{\frac{S}{N}}, 
\end{align*}
and thus \eqref{eq:convergence ind ass} holds at rank $n+1$ with $C_{n+1,r} = 2B_{r}+\widehat{C}_{n,r}$.

For part \ref{it:uniform convergence}) we have by Lemma \ref{lem:telescope reformulation}
\begin{align}
\E[|\pi^{N}_{n}(\varphi)-\pi_{n}(\varphi)|^{r}]^{\frac{1}{r}} 
\leq \sum_{p=0}^{n} \alpha_{p,n}\E\left[\left|\left(\pi^{N}_{p}-\Phi_{p}(\pi^{N}_{p-1})\right)\left(\varphi_{p,n,\varphi}\right)\right|^{r}\right]^{\frac{1}{r}} \label{eq:telescope decompo}
\leq \widehat{C}_{r}\frac{\delta}{\epsilon^3}\sqrt{\frac{S}{N}},
\end{align}
where the second inequality follows from Proposition \ref{prop:induction step}\eqref{it:step part b} and Lemma \ref{lem:telescope reformulation}.
\end{proof}

\section{Augmented resampling for particle islands}
\label{sec:another r2 resampler}

So far we have seen that by replacing the multinomial resampling (Algorithm \ref{alg:multinomial_resampling}) in the BPF with the augmented resampling (Algorithm \ref{alg:augmented_resampling}), we obtain a convergent approximation of $(\pi_{n})_{n\geq 0}$. However, the proposed algorithm has some shortcomings in efficiency which will address in this section.

First, we observe that at each stage of Algorithm \ref{alg:augmented_resampling}, each PE resamples $M$ particles out of $2M$ particles, which is in general more computationally expensive than resampling $M$ out of $M$ particles. Second, we observe that in order to do the resampling, a PE must receive the $M$ individual particle weights from the paired PE, which may imply a notable communication cost, especially for large $M$. In this section we propose a modification which addresses both of these sources of computation and communication cost; in the proposed method each PE resamples $M$ particles out of $M$ particles and communicates only a single weight with its paired PE at each stage.

The proposed modification is reminiscent to the IPF algorithm of \citet{verge_et_al15} with the exception that the between island (i.e.~between PE) resampling is done in multiple stages by means of augmented resampling. We dub the algorithm \emph{augmented island resampling particle filter (AIRPF)} and it is described in Algorithm \ref{alg: ir2pf} below, where we also use the shorthand notations, 
\begin{align*}
\widecheck{\vzeta}_{n} \defeq (\widecheck{\zeta}^{1}_{n},\ldots,\widecheck{\zeta}^{N}_{n})\qquad \text{ and }\qquad
\widecheck{\vW}_{n} \defeq (\widecheck{W}_{n}^{1},\ldots,\widecheck{W}_{n}^{N}), \qquad n\geq 0,
\end{align*}
where $\widecheck{\zeta}^{i}_{n}$ and $\widecheck{W}_{n}^{i}$ for all $1\leq i\leq N$ will be defined below by Algorithms \ref{alg: ir2pf} and \ref{alg:grouped mn}.


\begin{algorithm}[!ht]
\caption{Augmented Island Resampling Particle Filter}
\label{alg: ir2pf}
\begin{algorithmic}

\For{$i=1,\ldots,mM$}
	\State $\zeta^{i}_{0} \sim \pi_{0}$
\EndFor

\For{$n\geq 1$}
	\State $(\widecheck{\vzeta}_{n-1},\widecheck{\vW}_{n-1}) \leftarrow $\textsc{WithinIslandResample}$(\vzeta_{n-1},g_{n-1})$
	\State $\widecheck{g}_{n-1}(\,\cdot\,) \leftarrow \sum_{i=1}^{N}\widecheck{W}^{i}_{n-1}\ind[\,\cdot\, = \widecheck{\zeta}^{i}_{n-1}]$
	\State $\widehat{\vzeta}_{n-1} \leftarrow $\textsc{AugmentedIslandResample}$(\widecheck{g}_{n-1}, \widecheck{\vzeta}_{n-1})$

	\For{$i=1,\ldots,mM$}
		\State $\zeta^{i}_{n} \sim f(\widehat{\zeta}^{i}_{n-1},\,\cdot\,)$
	\EndFor
	
\EndFor	
\end{algorithmic}
\end{algorithm}

Essentially AIRPF has two resampling steps. First is the within island (i.e.~within PE) resampling step which preforms multinomial resampling of $M$ particles within each PE. Subsequently, in the second step, the $m$ groups of $M$ particles per PE are resampled by duplicating the entire samples of size $M$ without selecting individual particles within the samples. These two resampling subroutines will be analysed theoretically in the following two sections. The analysis is analogous to that conducted for the augmented resampling algorithm in Section \ref{sec:augmented analysis}. The more technical proof are postponed to Appendix \ref{sec:proof for modified algorithm}.

\subsection{Within island resampling}

For a formal description of \textsc{WithinIslandResample} we define $A \in \R^{mM\times mM}$ as
\begin{equation*}
A \defeq \Id_{m} \otimes \ones_{1/M},
\end{equation*}
and notations 
\begin{align*}
\vxi_{\mathrm{in}} \defeq (\xi^{1}_{\mathrm{in}},\ldots,\xi^{N}_{\mathrm{in}}),~\vW_{\mathrm{out}} \defeq (W^{1}_{\mathrm{out}},\ldots,W^{N}_{\mathrm{out}}),
~\text{ and }~
\vxi_{\mathrm{out}} \defeq (\xi^{1}_{\mathrm{out}},\ldots,\xi^{N}_{\mathrm{out}}).
\end{align*}
The within island resampling then proceeds as described in Algorithm \ref{alg:grouped mn}.
\begin{algorithm}[!ht]
\caption{Within island resample}
\label{alg:grouped mn}
\begin{algorithmic}[1]
\State {$(\vxi_{\mathrm{out}},\vW_{\mathrm{out}})=\textsc{WithinIslandResample}\left(\vxi_{\mathrm{in}},g\right)$}{}
\For{$i=1,\ldots,N$}
	\State $W^{i}_{\mathrm{out}} \leftarrow \sum_{j=1}^{N}A^{ij}g(\xi^{i}_{\mathrm{in}})$
	\State $\xi^{i}_{\mathrm{out}} \sim  (W^{i}_{\mathrm{out}})^{-1}\sum_{j=1}^{N}A^{ij}g(\xi^{i}_{\mathrm{in}})\delta_{\xi^{j}_{\mathrm{in}}}$ \label{it:withinisland cond dist}
\EndFor	

\end{algorithmic}
\end{algorithm}

From Algorithm \ref{alg:grouped mn} and the definition of $A$ we obtain the following expression for the  weights $\vW_{\mathrm{out}}$ returned by Algorithm \ref{alg:grouped mn}
\begin{equation}\label{eq:out weight}
W^{i}_{\mathrm{out}} \defeq \frac{1}{M}\sum_{j=1}^{M} g(\xi^{(k-1)M+j}_{\mathrm{in}}), \qquad (k-1)M < i \leq kM, \qquad 1\leq k \leq m.
\end{equation}
Note in particular that for any $1 \leq k \leq m$ the weights with indices in $\{(k-1)M+1,\ldots,kM\}$ are equal.

Proposition \ref{prop:gmn convergence} below is our main result for Algorithm \ref{alg:grouped mn}, and it is analogous to Proposition \ref{prop:intro_to_aug_resampling}; part \ref{it:part a}) establishes a result similar to Proposition \ref{prop:intro_to_aug_resampling} for the entire sample $\vxi_{\mathrm{out}}$ while part \ref{it:part b}) establishes a similar result for individual PEs, i.e.~the sub-samples $(\xi^{(k-1)M+1}_{\mathrm{out}},\ldots,\xi^{kM}_{\mathrm{out}})$, where $1 \leq k \leq m$. 

\begin{proposition}\label{prop:gmn convergence}
Assume \eqref{ass:NS}. If $\vxi_{\mathrm{in}}$ is any $\X^{N}$ valued random variable, $\vxi_{\mathrm{out}}$ is generated according to Algorithm \ref{alg:grouped mn}, $g\in\boundMeas_{+}(\X)$ and we define probability measures
\begin{equation}\label{eq:tmp measures}
\pi^{N} \defeq \dfrac{1}{N}\sum_{i=1}^{N}\delta_{\xi^{i}_{\mathrm{in}}}, \qquad \pi^{M,k} \defeq \dfrac{1}{M} \sum_{i=1}^{M} \delta_{\xi^{(k-1)M+i}_{\mathrm{in}}},\qquad 1 \leq k \leq m,
\end{equation}
and 
\begin{equation*}
\widecheck{\pi}^{N} \defeq \dfrac{\sum_{i=1}^{N}W^{i}_{\mathrm{out}}\delta_{\xi^{i}_{\mathrm{out}}}}{\sum_{i=1}^{N}W^{i}_{\mathrm{out}}} \qquad \widecheck{\pi}^{M,k} \defeq \dfrac{\sum_{i=1}^{M}W^{(k-1)M+i}_{\mathrm{out}}\delta_{\xi^{(k-1)M+i}_{\mathrm{out}}}}{\sum_{i=1}^{M}W^{(k-1)M+i}_{\mathrm{out}}}
\qquad 1 \leq k \leq m,
\end{equation*}
then 
\begin{enumerate}
\item\label{it:part a} there exists $B_{r} \in \R_{+}$ such that for any $\varphi \in \boundMeas(\X)$ 
\begin{align}\label{eq:gmn error}
\E\left[\left|\pi^{N}(g)\widecheck{\pi}^{N}(\varphi)-\pi^{N}(g\varphi)\right|^{r}\right]^{\frac{1}{r}} \leq \frac{B_{r}\infnorm{g}\osc{\varphi}}{\sqrt{N}},
\end{align}
\item\label{it:part b}
there exists $B_{r} \in \R_{+}$ such that for any $\varphi \in \boundMeas(\X)$ and any $1\leq k \leq m$ 
\begin{equation}\label{eq:gmn error b}
\lpnorm{r}{\pi^{M,k}(g)\left(\widecheck{\pi}^{M,k}(\varphi) - \frac{\pi^{M,k}(g\varphi)}{\pi^{M,k}(g)}\right)} \leq \frac{B_{r}\infnorm{g}\osc{\varphi}}{\sqrt{M}}.
\end{equation}
\end{enumerate}
\end{proposition}
\noindent
%

\subsection{Augmented island resampling}

Theoretically the augmented resampling for particle islands is very similar to the augmented resampling, Algorithm \ref{alg:augmented_resampling}. With appropriate notational conventions the theoretical analysis becomes nearly identical to that of Section \ref{sec:augmented analysis} with the exception that for any $1\leq s \leq S$, we replace individual particles $\xi^{i}_{s}$ by particle islands $$\vxi^{i}_{s} \defeq (\xi^{(i-1)M+1}_{s},\ldots,\xi^{iM}_{s})$$ and set $M=1$ to signify the fact that there is only one particle island per PE. Following the convention that $M=1$, we define matrices $\bA_{1},\ldots,\bA_{S}$ analogously to \eqref{eq:A matrices} as
\begin{equation*}
\bA_{s} \defeq \Id_{2^{S-s}} \otimes \ones_{1/2} \otimes \Id_{2^{s-1}}.
\end{equation*}
The resulting algorithm is described in Algorithm \ref{alg:new augmented resampling}.

\begin{algorithm}[!ht]
\caption{Augmented island resampling}
\label{alg:new augmented resampling}
\begin{algorithmic}[1]
\State {$\vxi_{\mathrm{out}}=\textsc{AugmentedIslandResample}\left(g,\vxi_{\mathrm{in}}\right)$}{}

\For{$i=1,\ldots,m$}
	\For{$j=1,\ldots,M$}
		\State $\xi^{(i-1)M+j}_{0} \leftarrow \xi^{(i-1)M+j}_{\mathrm{in}}$
	\EndFor
	\State $\bV^{i}_{0}\leftarrow M^{-1}\sum_{j=1}^{M}g(\xi^{(i-1)M+j}_{0})$
\EndFor

\For{$s=1,\ldots,S$}
	\For{$i=1,\ldots,m$}
		\State $\bV^{i}_{s} \leftarrow \sum_{j=1}^{m}\bA^{ij}_{s}\bV^j_{s-1}$
		\State $\vxi^i_{s}\sim (\bV^i_{s})^{-1}\sum_{j=1}^{m}\bA^{ij}_{s}\bV^j_{s-1}\delta_{\bxi^j_{s-1}}$
	\EndFor
\EndFor

\For{$i=1,\ldots,mM$}
	\State $\xi^{i}_{\mathrm{out}} \leftarrow \xi^{i}_{S}$.
\EndFor

\end{algorithmic}
\end{algorithm}

%
\begin{proposition}\label{prop:g augmented convergence}
Assume \eqref{ass:NS}. If $\vxi_{\mathrm{in}}$ is any $\X^N$ valued random variable, $\vxi_{\mathrm{out}}$ is computed according to Algorithm \ref{alg:new augmented resampling} and $g\in\boundMeas_{+}(\X)$, then for any $r\geq 1$ there exists $B_{r} \in \R_{+}$ such that for any $\varphi \in \boundMeas(\X)$, 
\begin{align}\label{eq:group augmented resample error}
\lpnorm{r}{\left(\frac{1}{m}\sum_{i=1}^{m}\vg(\vxi^{i}_{\mathrm{in}})\right)\left(\frac{1}{m}\sum_{i=1}^{m}\vphi(\vxi^{i}_{\mathrm{out}})\right) - \frac{1}{m}\sum_{i=1}^{m}\vg(\vxi^{i}_{\mathrm{in}})\vphi(\vxi^{i}_{\mathrm{in}})} 
\leq B_{r}\|g\|\sqrt{\frac{S}{m}}\osc{\varphi},
\end{align}
where 
\begin{align*}
\vg(\vx) \defeq \frac{1}{M}\sum_{i=1}^{M} g(x^{i}), \qquad \text{and} \qquad 
\vphi(\vx) \defeq \frac{1}{M}\sum_{i=1}^{M} \varphi(x^{i}), 
\end{align*}
for all $\vx = (x^1,\ldots,x^{M}) \in \X^{M}$.
\end{proposition}

Due to the similarity of the proof of Proposition \ref{prop:g augmented convergence} to that of Proposition \ref{prop:intro_to_aug_resampling} we will only outline the proof. First we construct a sequence $\overline{\cM} \defeq \{(\bX_{\rho},\bcF_{\rho}); 0 \leq \rho \leq Sm\}$ such that $\bX_{0} \defeq 0$ and $\bcF_{0} \defeq \sigma(\vxi_{\mathrm{in}})$ and for all $1 \leq \rho \leq Sm$ 
\begin{align*}
\bX_{\rho} &\defeq
\frac{\bV^{i_{m}(\rho)}_{s_{m}(\rho)}}{\sqrt{Sm}}\left(\vbarphi(\bxi^{i_{m}(\rho)}_{s_{m}(\rho)}) - \frac{\sum_{j=1}^{m} \bA^{i_{m}(\rho)j}_{s_{m}(\rho)}\bV^{j}_{s_{m}(\rho)-1}\vbarphi(\bxi^{j}_{s_{m}(\rho)-1})}{\bV^{i_{m}(\rho)}_{s_{m}(\rho)}} \right), \\
\bcF_{\rho} &\defeq \bcF_{\rho-1} \vee \sigma(\vxi^{i_{m}(\rho)}_{s_{m}(\rho)}),
\end{align*}
where $i_{m}(\rho)$ and $s_{m}(\rho)$ are as defined in \eqref{eq:index map}, and 
\begin{equation}\label{eq:centred phi bold}
\vbarphi(\vxi^{i}_{s}) \defeq \vphi(\vxi^{i}_{s}) - \frac{\sum_{j=1}^{m}\vg(\vxi^{j}_0)\vphi(\vxi^{j}_0)}{\sum_{j=1}^{m}\vg(\vxi^{j}_0)}, \qquad 1 \leq i \leq m,~1\leq s\leq S.
\end{equation}
With these notations we obtain the following result, analogous to Proposition \ref{prop:martingale decomposition}. The proof is essentially identical to that of Proposition \ref{prop:martingale decomposition} and hence omitted. 
\begin{proposition}\label{prop:martingale decomposition group}
Assume \eqref{ass:NS}. The following statements hold:
\begin{enumerate}

\item
$\bX_{\rho}$ is $\bcF_{\rho}$-measurable for all $0 \leq \rho \leq Sm$;

\item
$\E\left[\bX_{\rho}\,\mid\,\bcF_{\rho-1}\right] = 0$ (a.s.) for all $0 < \rho \leq Sm$;

\item
$\left|\bX_{\rho}\right| \leq \|g\|\osc{\varphi}/\sqrt{Sm}$ for all $0 \leq \rho \leq Sm$;

\item and we have the identities
\begin{align*}
\sqrt{\frac{S}{m}}\sum_{\rho=1}^{Sm} \bX_{\rho} = \frac{1}{m}\sum_{i=1}^{m} \bV^{i}_{S}\vbarphi(\vxi^{i}_{S}) 
= \left(\frac{1}{m}\sum_{i=1}^{m}\vg(\vxi^{i}_{0})\right)\left(\frac{1}{m}\sum_{i=1}^{m}\vphi(\vxi^{i}_{S})\right) - \frac{1}{m}\sum_{i=1}^{m} \vg(\vxi^{i}_{0})\vphi(\vxi^{i}_{0}) 
\end{align*}
\end{enumerate}
\end{proposition}
Proposition \ref{prop:g augmented convergence} then follows by Burkholder-Davis-Gundy inequality similarly as Proposition \ref{prop:intro_to_aug_resampling}.

\subsection{Modified augmented resampling}
\label{sec:modified augmentd resampling}

In the introduction we stated that augmented resampling enables a straightforward way to further control the communication of particle information between PEs. We will now address this claim more closely.

By Lemma \ref{lem:facts about A} we know that $\bA_{s}$ is symmetric and that for any $1\leq s \leq S$, each row of $\bA_{s}$ has exactly two nonzero elements of which one is on the diagonal. This implies that the pairs of indices of the nonzero columns for each row of $\bA_{s}$ form a partition of $\{1,\ldots,m\}$ into $m/2$ pairs of indices
\begin{equation*}
P_{i,s} \defeq (\ell^{i}_{s},r^{i}_{s}),\qquad 1  \leq i \leq m/2,
\end{equation*}
for which, by \eqref{eq:explicit nonzero elements}, we can obtain explicit expressions as
\begin{align*}
(\ell^{1}_{s},\ldots,\ell^{m/2}_{s}) &\defeq ((2^{s}(i-1) + (j-1) + 1)_{1 \leq j \leq 2^{s-1}})_{1 \leq i \leq 2^{S-s}}, \\
(r^{1}_{s},\ldots,r^{m/2}_{s}) &\defeq ((2^{s}(i-1) + (j-1) + 1 + 2^{s-1})_{1 \leq j \leq 2^{s-1}})_{1 \leq i \leq 2^{S-s}}.
\end{align*}
If, for any $1\leq s\leq S$ we associate the subsample $\vxi^{i}_{s}$ with PE $i$, as we have done so far, then the pair $P_{i,s}$ has the interpretation of representing the indices of PEs that are paired up for communication at stage $s$, and they are illustrated in Figure \ref{fig:pair_diagram}. 

\begin{remark}
For our purposes, the indexing of pairs $(P_{1,s},\ldots,P_{m/2,s})$ where $1\leq s \leq S$ is fixed could be replaced with any permutation of $\{1,\ldots,m/2\}$.
\end{remark}

%
%

\begin{figure}[t!]
\begin{center}

\tikzset{
    basicvertex/.style = {
    	draw,
		circle,
		minimum size = 1mm,
		outer sep = 0pt,
		fill=white,
		inner sep = 1pt
    }
}

\tikzset{
    basicedge/.style = {
		color = gray,
		shorten <=0cm, 
		shorten >=0cm,
		line width=.5pt    
		}
}

\mbox{
\begin{minipage}{0.47\textwidth}
\begin{center}
\mbox{

\begin{tikzpicture}
\tikzstyle{every node}=[scale=.8,font=\normalsize]
\def \vstep {-1.5}
\def \hstep {1}
\def \r {2}
\def \m {3} 
\def \N {8}
\def \X {4}
\def \Y {2}
\def \vspaces{{1.25,2,3,4}}

\foreach \i in {0,...,2}
{
	\foreach \j in {1,...,8}
	{	
		\foreach \k in {0,...,1}
		{
			\pgfmathsetmacro{\result}{\vspaces[\i]}
			\pgfmathsetmacro{\resulttwo}{\vspaces[\i+1]}
			\draw[basicedge]
			let 
				\n1 = {int(mod((\j-1),\r^(\i))+\k*\r^(\i)+\r^(\i+1)*int(floor((\j-1)/\r^(\i+1))))+1}, 
				\n2 = {\resulttwo}
			in 							
				(\j*\hstep,\result*\vstep) -- (\n1*\hstep,\n2*\vstep);
		}
	}
}

\draw[-,line width=2pt] (1*\hstep,2*\vstep) to [out=-45, in=225] node [midway,below=5pt,fill=white,inner sep=2pt,draw,line width=.5pt,rectangle] {$P_{1,1}$} (2*\hstep,2*\vstep);
\draw[-,line width=2pt] (3*\hstep,2*\vstep) to [out=-45, in=225] node [midway,below=5pt,fill=white,inner sep=2pt,draw,line width=.5pt,rectangle] {$P_{2,1}$} (4*\hstep,2*\vstep);
\draw[-,line width=2pt] (5*\hstep,2*\vstep) to [out=-45, in=225] node [midway,below=5pt,fill=white,inner sep=2pt,draw,line width=.5pt,rectangle] {$P_{3,1}$} (6*\hstep,2*\vstep);
\draw[-,line width=2pt] (7*\hstep,2*\vstep) to [out=-45, in=225] node [midway,below=5pt,fill=white,inner sep=2pt,draw,line width=.5pt,rectangle] {$P_{4,1}$} (8*\hstep,2*\vstep);

\draw[-,line width=2pt] (1*\hstep,3*\vstep) to [out=-45, in=225] node [midway,below=3pt,fill=white,inner sep=2pt,draw,line width=.5pt,rectangle] {$P_{1,2}$}(3*\hstep,3*\vstep) ;
\draw[-,line width=2pt] (2*\hstep,3*\vstep) to [out=-45, in=225] node [midway,below=3pt,fill=white,inner sep=2pt,draw,line width=.5pt,rectangle] {$P_{2,2}$}(4*\hstep,3*\vstep);
\draw[-,line width=2pt] (5*\hstep,3*\vstep) to [out=-45, in=225] node [midway,below=3pt,fill=white,inner sep=2pt,draw,line width=.5pt,rectangle] {$P_{3,2}$}(7*\hstep,3*\vstep);
\draw[-,line width=2pt] (6*\hstep,3*\vstep) to [out=-45, in=225] node [midway,below=3pt,fill=white,inner sep=2pt,draw,line width=.5pt,rectangle] {$P_{4,2}$}(8*\hstep,3*\vstep);

\draw[-,line width=2pt] (1*\hstep,4*\vstep) to [out=-30, in=210] node [midway,below=5pt,fill=white,inner sep=2pt,draw,line width=.5pt,rectangle] {$P_{1,3}$}(5*\hstep,4*\vstep);
\draw[-,line width=2pt] (2*\hstep,4*\vstep) to [out=-30, in=210] node [midway,below=5pt,fill=white,inner sep=2pt,draw,line width=.5pt,rectangle] {$P_{2,3}$}(6*\hstep,4*\vstep);
\draw[-,line width=2pt] (3*\hstep,4*\vstep) to [out=-30, in=210] node [midway,below=5pt,fill=white,inner sep=2pt,draw,line width=.5pt,rectangle] {$P_{3,3}$}(7*\hstep,4*\vstep);
\draw[-,line width=2pt] (4*\hstep,4*\vstep) to [out=-30, in=210] node [midway,below=5pt,fill=white,inner sep=2pt,draw,line width=.5pt,rectangle] {$P_{4,3}$}(8*\hstep,4*\vstep);

\foreach \i in {0,...,3}
{
	\foreach \j in {1,...,8}
	{
		\pgfmathparse{\vspaces[\i]};
		\pgfmathsetmacro{\s}{\i};
		\pgfmathsetmacro{\k}{\j};
	  	\node[basicvertex] (\i\j) at (\j*\hstep,\pgfmathresult*\vstep) {$\vxi^{\pgfmathprintnumber[precision=0]{\k}}_{\pgfmathprintnumber[precision=0]{\s}}$};
	}
}	

\end{tikzpicture}
}
\end{center}
\end{minipage}
}
\end{center}
\caption{Illustration of the paired PEs at each stage of the augemented island resampling in the case $m=8$.}
\label{fig:pair_diagram}
\end{figure}
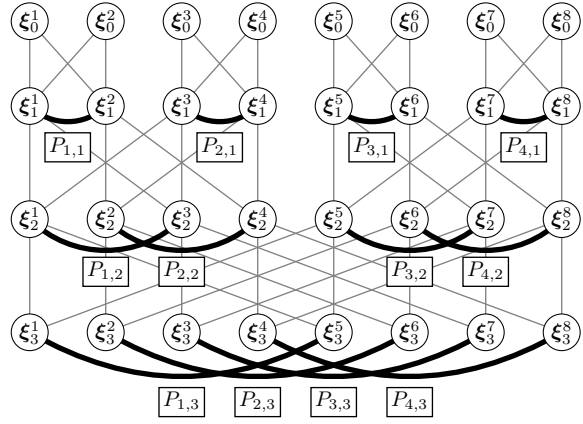


Now consider the PE $\ell^{i}_{s}$ at stage $s$ for some $1 \leq i \leq m/2$. By line \ref{line:sampling aug island} of Algorithm \ref{alg:new augmented resampling} and the discussion above we now see that when drawing the sample $\vxi^{\ell^{i}_{s}}_{s}$, PE $\ell^{i}_{s}$  essentially randomly decides whether to keep the sample $\vxi^{\ell^{i}_{s}}_{s-1}$ from the previous stage or assume the sample $\vxi^{r^{i}_{s}}_{s-1}$ of its paired PE. Simultaneously the paired PE $r^{i}_{s}$ makes randomly and independently a similar decision between $\vxi^{\ell^{i}_{s}}_{s-1}$ or $\vxi^{r^{i}_{s}}_{s-1}$.

It may be the case, in particular if $\bV^{\ell^{i}_{s}}_{s-1} \approx \bV^{r_{i,s}}_{s-1}$, that after the random sampling on line \ref{line:sampling aug island} of Algorithm \ref{alg:new augmented resampling} at stage $s$, one has $\big(\vxi^{\ell^{i}_{s}}_{s},\vxi^{r^{i}_{s}}_{s}\big) = \big(\vxi^{r^{i}_{s}}_{s-1},\vxi^{\ell^{i}_{s}}_{s-1}\big)$ i.e.~the paired PEs have simply exchanged their particles. Intuitively, it seems that performing this exchange is unnecessary, as the purpose of resampling is to duplicate particles appropriately many, possibly zero, times and hence only the number of duplicates is expected to matter, not the order in which they are allocated to the PEs. Thus to reduce the time spent on the communication between PEs, it seems advisable to avoid the above-mentioned exchange. 
\mbox{Algorithm \ref{alg:new augmented resampling mod}} describes a simple modification of Algorithm \ref{alg:new augmented resampling} designed to avoid this seemingly redundant particle exchange. 

%
%
%
%
%
\begin{algorithm}[!ht]
\caption{Modified Augmented Island Resampling}
\label{alg:new augmented resampling mod}
\begin{algorithmic}[1]
\label{alg:bf}
\State {$\vxi_{\mathrm{out}}=\textsc{AugmentedIslandResample}\left(g,\vxi_{\mathrm{in}}\right)$}{}

\For{$i=1,\ldots,m$}
	\For{$j=1,\ldots,M$}
		\State $\xi^{(i-1)M+j}_{0} \leftarrow \xi^{(i-1)M+j}_{\mathrm{in}}$
	\EndFor
	\State $\bV^{i}_{0}\leftarrow \vg(\vxi_{0}^{i})$
\EndFor

\For{$s=1,\ldots,S$}
	\For{$i=1,\ldots,m$}
		\State $\bV^{i}_{s} \leftarrow \sum_{j=1}^{m}\bA^{ij}_{s}\bV^j_{s-1}$
		\State $\widetilde{\vxi}^i_{s}\sim (\bV^i_{s})^{-1}\sum_{j=1}^{m}\bA^{ij}_{s}\bV^j_{s-1}\delta_{\bxi^j_{s-1}}$\label{line:sampling aug island}
	\EndFor

	\For{$i=1,\ldots,m/2$} \label{line:loop}
	\State\begin{equation*} 
\big(\vxi^{r^{i}_{s}}_{s},\vxi^{\ell^{i}_{s}}_{s}\big) = 
\begin{cases}
\big(\widetilde{\vxi}^{\ell^{i}_{s}}_{s},\widetilde{\vxi}^{r^{i}_{s}}_{s}\big), & \text{if}\quad \big(\widetilde{\vxi}^{r^{i}_{s}}_{s},\widetilde{\vxi}^{\ell^{i}_{s}}_{s}\big)=\big(\widetilde{\vxi}^{\ell^{i}_{s}}_{s-1},\widetilde{\vxi}^{r^{i}_{s}}_{s-1}\big)\\
\big(\widetilde{\vxi}^{r^{i}_{s}}_{s},\widetilde{\vxi}^{\ell^{i}_{s}}_{s}\big), &\text{otherwise.}
\end{cases}
\end{equation*}\label{line:exchange}
\EndFor
\EndFor

\For{$i=1,\ldots,mM$}
	\State $\xi^{i}_{\mathrm{out}} \leftarrow \xi^{i}_{S}$.
\EndFor

\end{algorithmic}
\end{algorithm}


The modification on lines \ref{line:loop} and \ref{line:exchange} of Algorithm \ref{alg:new augmented resampling mod} changes slightly the statistical behaviour of the algorithm and hence Propositions \ref{prop:g augmented convergence} and \ref{prop:martingale decomposition group} are not immediately valid for Algorithm \ref{alg:new augmented resampling mod}. However, similar results with an appropriate martingale difference construction can be obtained. 

We define a sequence $\widetilde{\cM} \defeq \{(\tX_{\rho},\tF_{\rho}); 0 \leq \rho \leq Sm/2\}$ such that $\tX_{0} = 0$ and $\tF_{0} = \sigma(\vxi_{0})$, and for all $0 < \rho \leq Sm/2$,
\begin{align}\label{eq:paired X}
\tX_{\rho} \defeq \tX^{r(\rho)}_{s(\rho)} + \tX^{\ell(\rho)}_{s(\rho)}, \qquad
\tF_{\rho} \defeq \tF_{\rho-1} \vee \sigma(\vxi^{r(\rho)}_{s(\rho)}) \vee \sigma(\vxi^{\ell(\rho)}_{s(\rho)}),
\end{align}
where $r(\rho) \defeq r^{i_{m/2}(\rho)}_{s_{m/2}(\rho)}$, $\ell(\rho) \defeq \ell^{i_{m/2}(\rho)}_{s_{m/2}(\rho)}$, $s(\rho) \defeq s_{m/2}(\rho)$ and $i_{m/2}(\rho)$ and $s_{m/2}(\rho)$ are as defined in \eqref{eq:index map}, and for all $1 \leq i\leq m/2$ and $1\leq s\leq S$
\begin{align*}
\widetilde{X}^{i}_{s} &= \frac{\bV^{i}_{s}}{\sqrt{Sm}}\left(\vbarphi(\vxi^{i}_{s}) - \frac{1}{\bV^{i}_{s}}\sum_{j=1}^{m}\bA_{s}^{ij}\bV^{j}_{s-1}\vbarphi(\vxi^{j}_{s-1})\right).
\end{align*}
Function $\vbarphi$ is as defined in \eqref{eq:centred phi bold}.


\begin{proposition}\label{prop:martingale for parsimonious communication}
Assume (A1). We have the following
\begin{enumerate}
\item $\tX_{\rho}$ is $\tF_{\rho}$-measurable for all $0\leq \rho \leq Sm/2$.
\item $\E[\tX_{\rho}\,\mid\,\tF_{\rho-1}] = 0$ a.s. for all $0 < \rho \leq Sm/2$.
\item $\tX_{\rho}\leq 2\|g\|\osc{\varphi}/\sqrt{Sm}$, for all $0 \leq \rho \leq Sm/2$.
\item and we have the identities
\begin{align}
\sqrt{\frac{S}{m}}\sum_{\rho=1}^{Sm/2} \tX_{\rho} &= \frac{1}{m}\sum_{i=1}^{m} \bV^{i}_{S}\vbarphi(\vxi^{i}_{S}) \label{eq:first decompo modified}\\
&= \left(\frac{1}{m}\sum_{i=1}^{m}\vg(\vxi^{i}_{0})\right)\left(\frac{1}{m}\sum_{i=1}^{m}\vphi(\vxi^{i}_{S})\right) - \frac{1}{m}\sum_{i=1}^{m} \vg(\vxi^{i}_{0})\vphi(\vxi^{i}_{0}) \label{eq:second decompo modified}
\end{align}

\end{enumerate}
\end{proposition}

\begin{proposition}\label{prop:island augmented resampling}
Assume \eqref{ass:NS} and let $g\in\boundMeas_{+}(\ss)$. Then for any $r\geq1$ there exists a finite constant $\widetilde{B}_{r}$, depending only on $r$, such that no matter what the distribution of $\vxi_{0}$ is, we have
\begin{align*}
\E\Bigg[\Bigg|\Bigg(\frac{1}{m}\sum_{i=1}^{m} \vg(\vxi_{\mathrm{in}}^i)\Bigg)\Bigg(\frac{1}{m}\sum_{i=1}^{m} \vphi(\vxi_{\mathrm{out}}^i)\Bigg)-\frac{1}{m}\sum_{i=1}^{m} \vg(\vxi_{\mathrm{in}}^i)\vphi(\vxi_{\mathrm{in}}^i)\Bigg|^{r}\Bigg]^{\frac{1}{r}} 
\leq \tilde{B_{r}}\sqrt{\frac{S}{m}} \|g\|\osc{\varphi}.
\end{align*}
\end{proposition}
\begin{proof}
The claim follows by applying the Burkholder-Davis-Gundy inequality to the expectation 
\begin{equation*}
\E\left[\left|\sqrt{\frac{S}{m}}\sum_{\rho=1}^{Sm/2}\tX_{\rho}\right|^{r}\right].
\end{equation*}
\end{proof}


\section{Augmented island resampling particle filter}
\label{sec:grouped r2pf}

We will now analyse the convergence properties of Algorithm \ref{alg: ir2pf}. The analysis is somewhat more complicated than the analogous analysis in Section \ref{sec:r2pf}. Additional complications arise due to Proposition \ref{prop:g augmented convergence} being independent of $M$. This implies that the two regimes identified earlier, i.e.~$m$ fixed, $M\to\infty$ (regime 1) and $m\to\infty$, $M$ fixed (regime 2), cannot be covered by one overarching analysis as before, but the scenarios have to be studied separately. The more technical proofs are postponed to Appendix \ref{sec:proof for airpf}.

\subsection{Convergence when $m$ is fixed and $M\to\infty$}

We introduce the following two PE specific empirical measure approximations
\begin{equation}\label{eq:PE-wise prediction appro}
\pi^{M,k}_{n} \defeq \frac{1}{M}\sum_{j=1}^{M} \delta_{\zeta^{(k-1)M+j}_{n}},\quad \widehat{\pi}^{M,k}_{n} \defeq \frac{1}{M}\sum_{j=1}^{M} \delta_{\widehat{\zeta}^{(k-1)M+j}_{n}},\quad 1 \leq k \leq m,
\end{equation}
for $\pi_{n}$ and $\widehat{\pi}_{n}$, respectively, based on the PE specific subsamples 
\begin{equation*}
\vzeta^{k}_{n}\defeq(\zeta^{(k-1)M+1}_{n},\ldots,\zeta^{kM}_{n}), \quad\text{and}\quad \widehat{\vzeta}^{k}_{n} \defeq (\widehat{\zeta}^{(k-1)M+1}_{n},\ldots,\widehat{\zeta}^{kM}_{n}) \qquad 1 \leq k \leq m.
\end{equation*}
With these notations we have the following analogue of Proposition \ref{prop:induction step}.

%
%
%
\begin{proposition}\label{prop:gr2ph induction step 1}
Assume \eqref{ass:NS} and \eqref{ass:potential}. If for some $n \geq 0$, there exists $C_{n,r}\in \R_{+}$ such that for all $1 \leq k \leq m$
\begin{equation}\label{eq:gr2pf induction ass}
\sup_{\varphi\in\boundMeas_{1}(\X)} \lpnorm{r}{\pi^{M,k}_{n}(\varphi) - \pi_{n}(\varphi)} \leq \frac{C_{n,r}}{\sqrt{M}},
\end{equation}
then there exists $\widehat{C}_{n,r} \in \R_{+}$ such that for all $1 \leq k \leq m$
\begin{align}\label{eq:claim inequality}
\sup_{\varphi\in\boundMeas_{1}(\X)} \lpnorm{r}{\widehat{\pi}^{M,k}_{n}(\varphi) - \widehat{\pi}_{n}(\varphi)} &\leq \frac{\widehat{C}_{n,r}}{\sqrt{M}}.
\end{align}
\end{proposition}
Proposition \ref{prop:gr2ph induction step 1} enables us to proof the convergence of Algorithm \ref{alg: ir2pf} by induction according to the following theorem.
\begin{theorem}\label{thm:grouped convergence 1}
Assume \eqref{ass:NS} and \eqref{ass:potential}. If $\pi^{N}_{n}$ is computed according to Algorithm \ref{alg: ir2pf}, then for all $n\geq 0$ there exists $C_{n,r}\in \R_{+}$ such that 
\begin{align*}
\sup_{\varphi\in\boundMeas_{1}(\X)}\lpnorm{r}{\pi^{N}_{n}(\varphi) - \pi_{n}(\varphi)} \leq \frac{C_{n,r}}{\sqrt{M}}.
\end{align*}
\end{theorem}
\begin{proof}
The proof is by induction, the assumption being that for some $n \geq 0$
\begin{equation}\label{eq:gr2pf induction assumption}
\sup_{\varphi\in\boundMeas_{1}(\X)}\sup_{1\leq k \leq m} \lpnorm{r}{\pi^{M,k}_{n}(\varphi)-\pi_{n}(\varphi)}\leq \frac{C_{n,r}}{\sqrt{M}}.
\end{equation}
The induction is started by observing that \eqref{eq:gr2pf induction assumption} holds for $n=0$ by Lemma \ref{lem:delmoral}. Now suppose \eqref{eq:gr2pf induction assumption} holds for some $n\geq 0$. By Minkowski's inequality
\begin{align*}
\sup_{\varphi\in\boundMeas_{1}(\X)}\lp{r}{\pi^{M,k}_{n+1}(\varphi) - \pi_{n+1}(\varphi)} &\leq 
\sup_{\varphi\in\boundMeas_{1}(\X)}\lp{r}{\pi^{M,k}_{n+1}(\varphi) - \widehat{\pi}^{M,k}_{n}(f(\varphi))} \\&~+ 
\sup_{\varphi\in\boundMeas_{1}(\X)}\lp{r}{\widehat{\pi}^{M,k}_{n}(f(\varphi))-\widehat{\pi}_{n}(f(\varphi))}.
\end{align*}
Similarly as in the proof of Theorem \ref{thm:convergence} we can bound the two terms on the r.h.s.~by using Lemma \ref{lem:mutation martingale} and Proposition \ref{prop:gr2ph induction step 1}, respectively, to see that \eqref{eq:gr2pf induction assumption} holds for $n+1$ with $C_{n+1,r} = 2B_{r} + \widehat{C}_{n,r}$.
\end{proof}

\subsection{Convergence when $m \to \infty$ and $M$ is fixed}

For regime 2 we have the following analogues of Proposition \ref{prop:gr2ph induction step 1} and Theorem \ref{thm:grouped convergence 1}.
\begin{proposition}\label{prop:induction step M fixed}
Assume \eqref{ass:NS} and \eqref{ass:potential}. If for some $n\geq 0$ there exists $C_{n,r}\in \R_{+}$ such that
\begin{equation}\label{eq:group convergence precondition}
\sup_{\varphi\in\boundMeas_{1}(\X)}\lpnorm{r}{\pi^{N}_{n}(\varphi) - \pi_{n}(\varphi)} \leq C_{n,r}\sqrt{\frac{S}{m}},
\end{equation} 
then there exists $\widehat{C}_{n,r}\in \R$ such that
\begin{equation*}
\sup_{\varphi\in\boundMeas_{1}(\X)} \lpnorm{r}{\widehat{\pi}^{N}_{n}(\varphi) - \widehat{\pi}_{n}(\varphi)} \leq \widehat{C}_{n,r}\sqrt{\frac{S}{m}}, 
\end{equation*}
where $\widehat{\pi}^{N}_{n} = N^{-1}\sum_{i=1}^{N}\delta_{\widehat{\xi}_{n}^{i}}$ for all $n\geq 0$.
\end{proposition}

\begin{theorem}\label{thm:g convergence 2} 
Assume \eqref{ass:NS} and \eqref{ass:potential}. If $\pi^{N}_{n}$ is computed according to Algorithm \ref{alg: ir2pf}, then for all $n\geq 0$ there exists $C_{n,r} \in \R_{+}$ such that 
\begin{align*}
\sup_{\varphi\in\boundMeas_{1}(\X)}\lpnorm{r}{\pi^{N}_{n}(\varphi) - \pi_{n}(\varphi)} \leq C_{n,r}{\sqrt{\frac{S}{m}}}.
\end{align*}
\end{theorem}
\begin{proof}
The proof follows by induction analogously to the proof of Theorem \ref{thm:grouped convergence 1} by using Lemmata \ref{lem:delmoral} and \ref{lem:mutation martingale} and Proposition \ref{prop:induction step M fixed}.
\end{proof}

\begin{remark}
We can extend Proposition \ref{prop:induction step M fixed} and Theorem \ref{thm:g convergence 2} straightforwardly to Algorithm \ref{alg: ir2pf} deploying the modified augmented resampling algorithm (Algorithm \ref{alg:new augmented resampling mod}) presented in Section \ref{sec:modified augmentd resampling}. This follows from observing that at every step of the respective proofs we can replace Proposition \ref{prop:g augmented convergence} with Proposition \ref{prop:island augmented resampling}.
\end{remark}

\section{Numerical experiments}
\label{sec:numerics}

We have seen that AIRPF, as well as the BPF deploying augmented resampling are valid convergent algorithms, but we have also seen in Theorems \ref{thm:convergence} and \ref{thm:g convergence 2} that by imposing constraints on the interactions we also introduce error, which manifests itself as slower convergence rate. This raises the practically important question whether these algorithms are not only faster than the existing methods, but is the speedup significant enough to outweigh the introduced error. We now aim to shed some light to this question with numerical experiments.

\subsection{Experimental setup}

In order to obtain accurate error estimates, we chose to run the experiments on a simple random walk HMM which admits exact numerical calculation by Kalman filter recursions \citep{kalman60}. The model we used is
\begin{align*}
\arraycolsep=1.4pt
\begin{array}{rlrll}
X_{0} &\sim \cN(\zeros_{d},\Id_{d}),&&\\
X_{n} &= X_{n-1} + \eta_{n},\qquad&\eta_{n} &\sim \cN(\zeros_{d},\Id_{d}),\quad & \qquad n>0, \\
Y_{n} &= X_{n} + \varepsilon_{n},&  \varepsilon_{n} &\sim \cN\!\left(\zeros_{d},\dfrac{1}{4}\Id_{d}\right), & \qquad n\geq 0,
\end{array}
\end{align*}
where $\zeros_{d}$ denotes a vector of zeros in $\R^{d}$ and $\cN(\mu,\sigma)$ denotes a Gaussian distribution with mean $\mu$ and covariance $\sigma$. 

The approximation error was taken to be the mean squared error
\begin{equation*}
\mathrm{MSE} \defeq \frac{1}{J}\sum_{j=1}^{J}
\sum_{n=0}^{n_{\mathrm{max}}-1}\sum_{i=1}^{d}\left(\overline{x}^{N,i}_{n,j} - \overline{x}^{i}_{n}\right)^2.
\end{equation*}
where $J$ is the number of independent runs, $n_{\mathrm{max}}$ is the length of the time series, $\overline{x}_{n} \defeq (\overline{x}^{1}_{n},\ldots,\overline{x}^{d}_{n})$ is the mean of the true filtering distribution, and $\overline{x}^{N}_{n,j} \defeq (\overline{x}^{N,1}_{n,j},\ldots,\overline{x}^{N,d}_{n,j})$, where $1\leq j \leq J$, denotes the approximation of $\overline{x}_{n}$ at the $j$th run.

We used this model with $d=7$, to generate a data set of length $n_{\mathrm{max}}=8000$ iterations and the error was calculated of $J=5$ independent runs. The choice of dimension $d=7$ is largely arbitrary although very low dimensions were intentionally avoided to introduce some pressure for the ESS to take low values which in turn emphasises the role of adaptive resampling scheme.

The algorithms were implemented in C and the parallelism was implemented using Intel MPI. The experiments were conducted on the high performance computing system Balena at the University of Bath using 16 computing nodes each capable of running 16 processes simultaneously. The code is available at \url{https://github.com/heinekmp/AIRPF}

\subsection{Algorithms}

For reference, two versions of IPF were implemented. The first version was our implementation of the original IPF which performed the between island resampling first and the between island resampling second (IPF1). A slight gain in efficiency is expected if the order of these resampling steps is reversed as this would mean that PEs would not have to communicate the individual particles but only a single weight per PE. The IPF with this reversed resampling order we call IPF2. 

For the algorithms proposed in this paper, we implemented AIRPF with the modified augmented island resampling algorithm, Algorithm \ref{alg:new augmented resampling mod} (AIRPF1). In accordance with our discussion in Section \ref{sec:r2pf intro}, we also implemented AIRPF with the fully adapted resampling scheme (AIRPF2). In order to make the comparison against IPF as fair as possible, also both versions of IPF deployed a modification analogous to Algorithm \ref{alg:new augmented resampling mod}; if the sample of any PE was duplicated at the between islands resampling stage, then the PE in question was ensured to keep one copy of the sample set.

\subsection{Results}

The algorithms were run with 22 roughly equally spaced values of  $M$ in the interval $\{200,\ldots,4200\}$ and $m \in \{64,128,256\}$. The resampling threshold $\theta$ took values in $\{.1,.2,.4,.6,.8,1\}$. Each computing node used in our experiments always used its full capacity of 16 processes. 

If we fix $m$ and $M$ and let $\theta$ vary, we obtain four MSE vs.~time curves; one curve for each algorithm. Figure \ref{fig:cleanup_demo}a illustrates such sets of four curves for four different values of $M\in \{200,400,600,800\}$. To clarify which curves are obtained with the same value of $M$, the curves obtained with $M=200$ are highlighted by a rectangle in Figure \ref{fig:cleanup_demo}a. In order to improve the visualisation by reducing the overlap of the curves, we calculated the lower envelope curves for each algorithm as shown in Figure \ref{fig:cleanup_demo}b. The horizontal dashed line at level 2396 denotes the worst case MSE which is obtained by taking the raw observations as the estimates of the filtering mean.


 \begin{figure}
 \includegraphics[width=\textwidth]{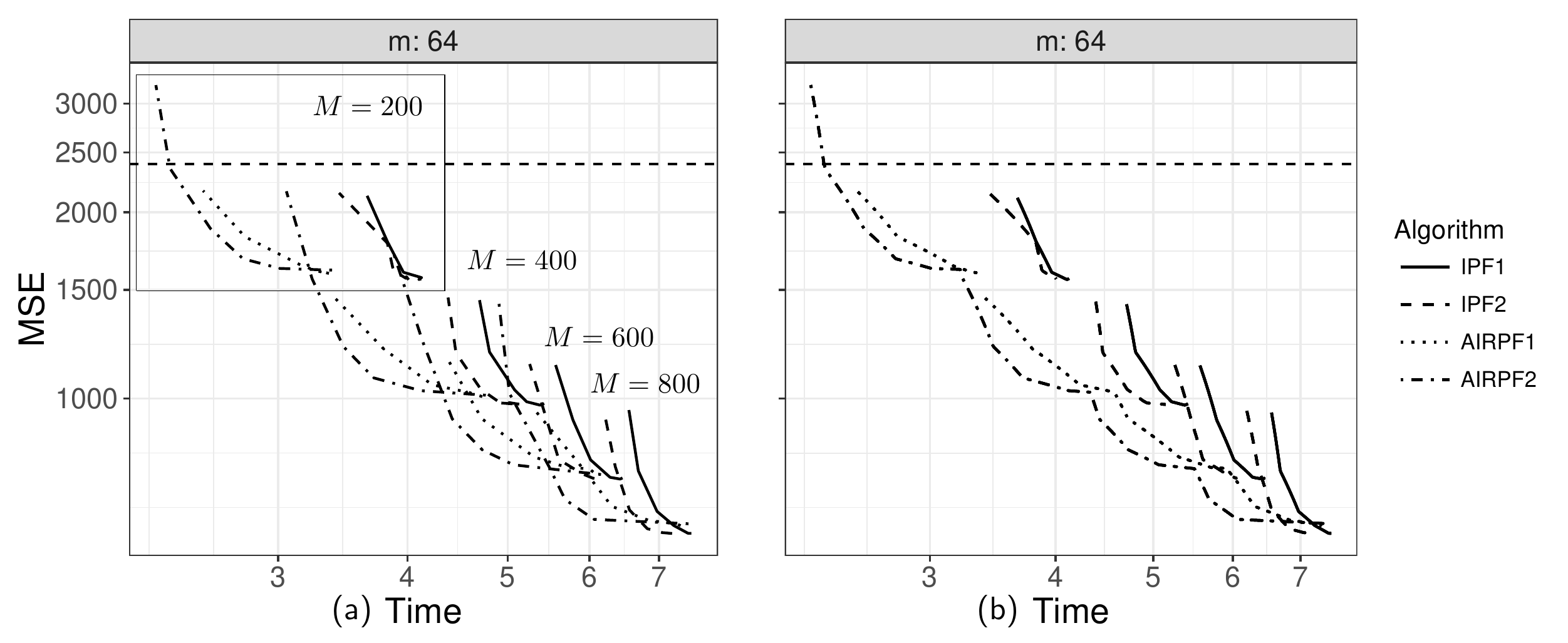}
 \caption{(a) MSE vs.~time plots for $M \in \{200,400,600,800\}$, $m=64$, and $\theta = \{.1,.2,.4,.6,.8,1\}$. For $M=200$ the curves obtained with different algorithms are highlighted by a rectangle. The dashed horizontal line at 2396 represents the raw MSE obtained by using the plain observations to approximate the mean of the filtering distribution. (b) The lower envelopes of the curves in (a).}
 \label{fig:cleanup_demo}
 \end{figure}

Figure \ref{fig:summary}a shows the lower envelopes of MSE vs.~time curves for the entire range of $M$ and $m$. Differences between IPF and AIRPF are more pronounced for larger values of $m$ and, in particular, for moderate values of $M$. For large $M$, the  differences vanish as the resampling that takes place within each PE independently begins to dominate the execution time. For moderate values of $M$, communication cost plays a more significant role, and in this case, AIRPF is more efficient than IPF. Also a notable gain in performance can be observed due to the fully adapted resampling. Figure \ref{fig:summary}b summarises the lower envelopes for AIRPF2 and IPF2 for the whole range of $M$ and $m$.


 \begin{figure}
 \includegraphics[width=\textwidth]{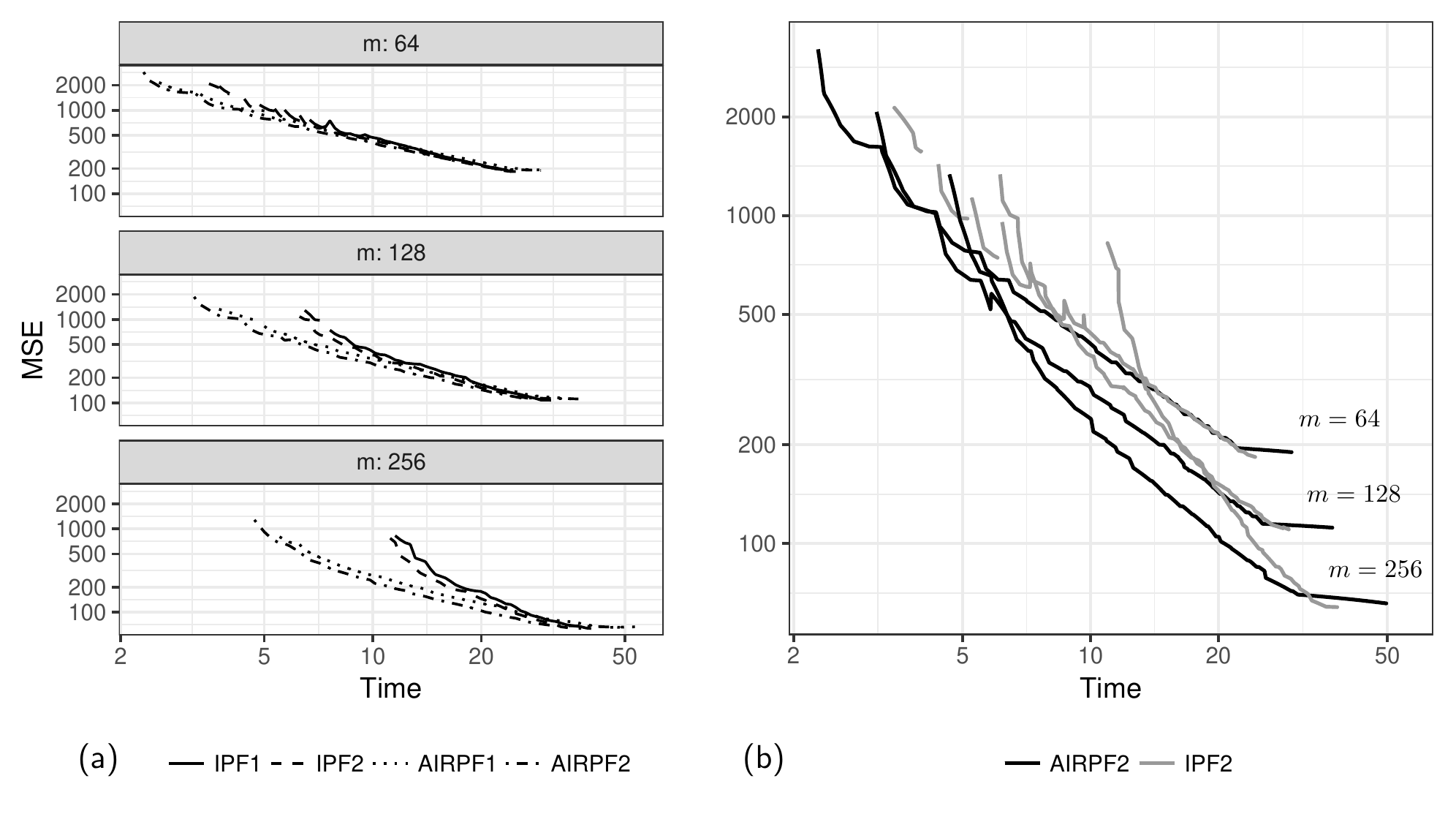}
 \caption{(a) The lower envelopes of MSE vs.~time for all algorithms and entire ranges of $M$ and $m$. (b) The lower envelopes of MSE vs.~time for AIRPF2 and IPF2 and entire ranges of $M$ and $m$.}
 \label{fig:summary}
 \end{figure}

\subsection{Conclusions}

Based on the results reported above, it appears that AIRPF shows the best potential in scenarios where the runtime is critical but the best possible accuracy with given resources is required. In such a case, increasing $M$ is not an option as it implies longer computation time. Also, the speedup by increasing the computer clock rate, which in turn would enable larger values of $M$, is not an option as modern computers have essentially reached their limit in clock rate. Therefore the only option is to increase $m$ in which case our experiments suggest that AIRPF could be the method of choice. 

Both AIRPF1 and AIRPF2 use the modification proposed in Section \ref{sec:modified augmentd resampling}. Without this modification the performance of AIRPF would have been significantly worse and hence the efficiency of AIRPF can be largely attributed to the ideas presented in Section \ref{sec:modified augmentd resampling}. 

We believe that the performance of AIRPF can be further improved in two specific scenarios. The first scenario is the following simple algorithmic modification. In the current fully adapted AIRPF the augmented resampling always begins at stage $s=1$, but presumably fewer resampling stages would have to be executed in total if the resampling was started at $s+1$, $s$ being the last executed resampling stage of the previous iteration that included resampling. The rationale for this modification is simple. By starting the resampling always at stage $s=1$ we introduce a bias towards the pairwise interactions associated with stage $s=1$ but by rotating the first resampling stage this bias is removed and more complete interactions are obtained which in turn is expected to increase ESS and lead to fewer resampling stages being executed in total.

The other scenario which seems particularly well suited for AIRPF is a computer network with a hypercube topology which matches exactly the radix-2 butterfly diagram structure of the AIRPF resampling step. We believe that a computing system based on such network topology would resemble the idealised computing system, discussed in Section \ref{sec:r2pf intro}, more accurately and hence make the reasoning, also presented in Section \ref{sec:r2pf intro} more valid. The experiments presented above were not executed in a hypercube architecture and hence the performance of AIRPF2 is mostly attributed to the ideas presented in Section \ref{sec:modified augmentd resampling}, as stated above.

For future theoretical research the convergence properties of the fully adapted resampling AIRPF remain to be analysed, although our conjecture is that similar uniform convergence results as in \citep{whiteley_et_al14} by controlling the ESS can be obtained. We finish on a more practical note by pointing out that the validity of the proposed algorithms will hold for various definitions of $A_{1},\ldots,A_{S}$. In the context of the present paper they define an interaction pattern which corresponds to a hypercube topology. Different definitions will lead to different interaction patterns that may have counterparts in computer architectures and may admit particularly efficient implementations.

\vspace{2mm}

This research made use of the Balena High Performance Computing (HPC) Service at the University of Bath.

\bibliography{mybib.bib}

\begin{thebibliography}{25}
\expandafter\ifx\csname natexlab\endcsname\relax\def\natexlab#1{#1}\fi
\expandafter\ifx\csname url\endcsname\relax
  \def\url#1{\texttt{#1}}\fi
\expandafter\ifx\csname urlprefix\endcsname\relax\def\urlprefix{URL: }\fi

\bibitem[{Boli\'c et~al.(2005)Boli\'c, Djuri\'c and Hong}]{bolic_et_al_05}
Boli\'c, M., Djuri\'c, P.~M. and Hong, S. (2005) Resampling algorithms and
  architectures for distributed particle filters.
\newblock \textit{IEEE TRANSACTIONS ON SIGNAL PROCESSING}, \textbf{53},
  2442--2450.

\bibitem[{Chopin(2004)}]{smc:the:C04}
Chopin, N. (2004) Central limit theorem for sequential {Monte Carlo} methods
  and its application to {Bayesian} inference.
\newblock \textit{Ann. Statist.}, \textbf{32}, 2385--2411.

\bibitem[{Crisan and Doucet(2002)}]{crisan_et_doucet02}
Crisan, D. and Doucet, A. (2002) A survey of convergence results on particle
  filtering methods for practitioners.
\newblock \textit{IEEE Transactions On Signal Processing}, \textbf{50},
  736--746.

\bibitem[{{Del Moral}(2004)}]{delmoral04}
{Del Moral}, P. (2004) \textit{{Feynman}-{Kac} Formulae. Genealogical and
  interacting particle systems with applications}.
\newblock Probability and its Applications. New York: Springer Verlag.

\bibitem[{{Del Moral} and Guionnet(1999)}]{del1999central}
{Del Moral}, P. and Guionnet, A. (1999) Central limit theorem for nonlinear
  filtering and interacting particle systems.
\newblock \textit{Ann. Appl. Probab.}, \textbf{9}, 275--297.

\bibitem[{{Del Moral} and Guionnet(2001)}]{delmoral_et_guionnet01}
--- (2001) On the stability of interacting processes with applications to
  filtering and genetic algorithms.
\newblock \textit{Ann. Inst. H. Poincar{\'e}, Probab. et Stat}, \textbf{37}.

\bibitem[{{Del Moral} et~al.(2017){Del Moral}, Moulines, Olsson and
  Verg{\'e}}]{delmoral17}
{Del Moral}, P., Moulines, E., Olsson, J. and Verg{\'e}, C. (2017) Convergence
  properties of weighted particle islands with application to the double
  bootstrap algorithm.
\newblock \textit{Stochastic Systems}, \textbf{6}, 367--419.

\bibitem[{Douc et~al.(2014)Douc, Moulines and Olsson}]{douc_et_al14}
Douc, R., Moulines, E. and Olsson, J. (2014) Long-term stability of sequential
  {M}onte {C}arlo methods under verifiable conditions.
\newblock \textit{Ann. Appl. Probab.}, \textbf{24}.

\bibitem[{Doucet et~al.(2001)Doucet, {De Freitas} and Gordon}]{doucet_et_al01}
Doucet, A., {De Freitas}, N. and Gordon, N. (eds.) (2001) \textit{Sequential
  {Monte Carlo} methods in practice}.
\newblock New York: Springer.

\bibitem[{Gordon et~al.(1993)Gordon, Salmond and Smith}]{gordon_et_al_93}
Gordon, N.~J., Salmond, D.~J. and Smith, A. F.~M. (1993) Novel approach to
  nonlinear/non-gaussian bayesian state estimation.
\newblock \textit{IEE PROCEEDINGS-F}, \textbf{140}.

\bibitem[{Heine and Whiteley(2016)}]{heine_et_whiteley16}
Heine, K. and Whiteley, N. (2016) Fluctuations, stability and instability of a
  distributed particle filter with local exchange.
\newblock \textit{Stochastic Processes and their Applications}, \textbf{127},
  2508--2541.

\bibitem[{Heine et~al.(2014)Heine, Whiteley, Cemgil and
  G\"ulda\c{s}}]{heine_et_al_14}
Heine, K., Whiteley, N., Cemgil, A.~T. and G\"ulda\c{s}, H. (2014) Butterfly
  resampling: Asymptotics for particle filters with constrained interactions.
\newblock \textit{arXiv}.

\bibitem[{Horn and Johnson(1991)}]{horn_et_johnson}
Horn, R.~A. and Johnson, C.~R. (1991) \textit{Topics in matrix analysis}.
\newblock Cambridge University Press.

\bibitem[{Kalman(1960)}]{kalman60}
Kalman, R.~E. (1960) A new approach to linear filtering and prediction
  problems.
\newblock \textit{ournal of Basic Engineering}, \textbf{82}.

\bibitem[{Lee et~al.(2010)Lee, Yau, Doucet and Holmes}]{lee_et_al10}
Lee, A., Yau, C., Doucet, A. and Holmes, C. (2010) On the utility of graphics
  cards to perform massively parallel simulation of advanced {M}onte {C}arlo
  methods.
\newblock \textit{J. Comput. Graph. Statist.}, \textbf{19}.

\bibitem[{Liu and Chen(1995)}]{liu_et_chen95}
Liu, J.~S. and Chen, R. (1995) Blind deconvolution via sequential imputation.
\newblock \textit{Journal of the American Statistical Association},
  \textbf{90}.

\bibitem[{Liu and Chen(1998)}]{liu_et_chen98}
--- (1998) Sequential {M}onte {C}arlo methods for dynamic systems.
\newblock \textit{Journal of the American Statistical Association},
  \textbf{93}.

\bibitem[{M\'iguez(2007)}]{miguez07}
M\'iguez, J. (2007) Analysis of parallelizable resampling algorithms for
  particle filtering.
\newblock \textit{Signal Process.}, \textbf{87}, 3155--3174.

\bibitem[{M\'iguez(2014)}]{miguez14}
--- (2014) On the uniform asymptotic convergence of a distributed particle
  filter.
\newblock In \textit{Proc. of the 8th Sens. Array and Multichannel Signal
  Process. Workshop (SAM 2014)}, 241--244. IEEE.

\bibitem[{M\'iguez and {V\'azquez}(2015)}]{miguez_et_vazquez15}
M\'iguez, J. and {V\'azquez}, M.~A. (2015) A proof of uniform convergence over
  time for a distributed particle filter.
\newblock arXiv:1504.01079v1.

\bibitem[{Murray et~al.(2016)Murray, Lee and Jacob}]{murray_et_al16}
Murray, L., Lee, A. and Jacob, P. (2016) Parallel resampling in the particle
  filter.
\newblock \textit{J. Comput. Graph. Statist.}, \textbf{25}.

\bibitem[{Oppenheim(1975)}]{oppenheim75}
Oppenheim, A. (1975) \textit{Digital Signal Processing}.
\newblock Prenctice-Hall.

\bibitem[{Pacheco(2011)}]{pacheco2011}
Pacheco, P. (2011) \textit{An Introduction to Parallel Programming}.
\newblock Morgan Kaufmann Publishers Inc., 1st edn.

\bibitem[{Verg{\'e} et~al.(2015)Verg{\'e}, Dubarry, {Del Moral} and
  Moulines}]{verge_et_al15}
Verg{\'e}, C., Dubarry, C., {Del Moral}, P. and Moulines, E. (2015) On parallel
  implementation of sequential monte carlo methods: the island particle model.
\newblock \textit{Statistics and Computing}, \textbf{25}, 243--260.

\bibitem[{Whiteley et~al.(2014)Whiteley, Lee and Heine}]{whiteley_et_al14}
Whiteley, N., Lee, A. and Heine, K. (2014) On the role of interaction in
  sequential monte carlo algorithms.
\newblock \textit{Bernoulli}.

\end{thebibliography}

\newpage
\begin{appendices}
%
\section{Proofs for Section \ref{sec:augmented analysis}}
\label{sec:proofs augmented}

\begin{proof}[Proof of Lemma \ref{lem:facts about A}] 
First we recall the \emph{mixed product property} of Kronecker product: for any matrices $A$, $B$, $C$ and $D$, such that the products $AC$ and $BD$ are defined, one has (see, e.g.~\cite{horn_et_johnson})
\begin{equation}\label{eq:mixed product property}
(A \otimes B)(C\otimes D)=(AC)\otimes (BD).
\end{equation}
Also we note that for any two square matrices $A \in \R^{p\times p}$ and $B \in \R^{q\times q}$ we have the element-wise formula:
\begin{equation}\label{eq:elementwise Kronecker formula}
(A\otimes B)^{ij}=A^{\floor{\frac{i-1}{q}}+1,\floor{\frac{j-1}{q}}+1} B^{((i-1)\bmod q)+1,((j-1)\bmod q)+1},
\end{equation}
where $i,j\in\{1,\ldots,pq\}$. By \eqref{eq:elementwise Kronecker formula}, $A\otimes B$ is symmetric whenever $A$ and $B$ are symmetric, proving the symmetry. Associativity of the Kronecker product and repeated applications of \eqref{eq:mixed product property} to the definition of $A_{s}$ in \eqref{eq:A matrices} yield $A_{s}A_{s}=A_{s}$ proving the idempotence. Also, by associativity and repeated applications of \eqref{eq:elementwise Kronecker formula} to \eqref{eq:A matrices}, we have 
\begin{align}
A_{s}^{ij} &= \Id_{2^{S-s}}^{\floor{\frac{i-1}{2^sM}}+1,\floor{\frac{j-1}{2^sM}}+1}\ones_{1/2}^{\left(\floor{\frac{i-1}{2^{s-1}M}}\bmod 2\right)+1, \left(\floor{\frac{j-1}{2^{s-1}M}}\bmod 2\right)+1} \nonumber\\ 
&\times\Id_{2^{s-1}}^{\left(\floor{\frac{i-1}{M}}\bmod 2^{s-1}\right)+1, \left(\floor{\frac{j-1}{M}}\bmod 2^{s-1}\right)+1}\Id_{1/M}^{((i-1) \bmod M )+1, ((j-1)\bmod )+1}.	\label{eq:elementwise formula for radix}
\end{align}
From this we see immediately that $A_{s}^{ij}\in\{0,(2M)^{-1}\}$.

By the idempotence, symmetry and the facts that by \eqref{eq:elementwise formula for radix}, $A_{s}^{ii}=(2M)^{-1}$ and $A_{s}^{ij}\in\{0,(2M)^{-1}\}$ one has
\begin{equation*}
\frac{1}{2M} = A_{s}^{ii} = (A_{s}A_{s})^{ii} = (A_{s}^TA_{s})^{ii} = \sum_{j=1}^{mM} (A^{ij}_{s})^2 = \frac{u}{(2M)^2} \iff u = 2M,
\end{equation*}
where $u$ is the number of non-zero elements on the $i$th column of $A_{s}$. Hence double stochasticity follows by symmetry.

To prove \eqref{eq:explicit nonzero elements} we observe that by setting $M=1$ in \eqref{eq:elementwise formula for radix} we have 
\begin{align*}
(\Id_{2^{S-s}} \otimes \ones_{1/2} \otimes \Id_{2^{s-1}})^{ij} &= \Id_{2^{S-s}}^{\floor{\frac{i-1}{2^s}}+1,\floor{\frac{j-1}{2^s}}+1}\ones_{1/2}^{\left(\floor{\frac{i-1}{2^{s-1}}}\bmod 2\right)+1, \left(\floor{\frac{j-1}{2^{s-1}}}\bmod 2\right)+1} \\&\times\Id_{2^{s-1}}^{\left((i-1)\bmod 2^{s-1}\right)+1, \left((j-1)\bmod 2^{s-1}\right)+1}.
\end{align*}
From this, by considering only the diagonal elements of the identity matrices, we have readily that the indices of the nonzero columns of the $i$th row of $\Id_{2^{S-s}} \otimes \ones_{1/2} \otimes \Id_{2^{s-1}}$ are those $1 \leq j \leq m$ for which 
\begin{align*}
\bigg\lfloor\frac{i-1}{2^{s}}\bigg\rfloor = \bigg\lfloor\frac{j-1}{2^{s}}\bigg\rfloor,\quad\text{ and } \quad\big((i-1) \bmod 2^{s-1}\big) = \big((j-1) \bmod 2^{s-1}\big).
\end{align*}
To prove that this is a superset of \eqref{eq:explicit nonzero elements}, suppose that 
\begin{equation}
j = \big((i-1) \bmod 2^{s-1}\big) + (q-1)2^{s-1} + 2^s\big\lfloor{(i-1)}/{2^s}\big\rfloor+1, \quad q \in \{1,2\}.
\label{eq:choice of beta}
\end{equation}
It is then simple to check that $\lfloor{(j-1)}/{2^s}\rfloor = \lfloor{(i-1)}/{2^s}\rfloor$ and 
\begin{equation*}
\big((j-1) \bmod 2^{s-1}\big) = j-1 - \bigg\lfloor{\frac{j-1}{2^{s-1}}}\bigg\rfloor 2^{s-1} = \big((i-1) \bmod 2^{s-1}\big). 
\end{equation*}
To prove the converse inclusion, suppose that $\lfloor{(i-1)}/{2^s}\rfloor = \lfloor{(j-1)}/{2^s}\rfloor$ and $\big((i-1) \bmod 2^{s-1}\big) = \big((j-1) \bmod 2^{s-1}\big)$.
Then one can check that
\begin{equation*}
j-1 = \big((i-1) \bmod 2^{s-1}\big) + 2^{s}\floor{\frac{i-1}{2^{s}}} + 2^{s-1}\left(\floor{\frac{j-1}{2^{s-1}}}\bmod 2\right), 
\end{equation*}
and since $\big(\lfloor{(j-1)}/{2^{s-1}}\rfloor\bmod 2\big) + 1\in \{1,2\}$, the claim follows.
\end{proof}

\begin{proof}[Proof of Lemma \ref{lem:ones product}]
We prove by induction that 
\begin{align}\label{eq:induction form}
\prod_{s=1}^{k} A_{s} = \Id_{2^{S-k}} \otimes \ones_{1/2^k} \otimes \ones_{1/M}
\end{align}
holds for all $1 \leq k \leq S$. Case $k=S$ then yields the claim. For $k=1$, \eqref{eq:induction form} holds by definition. Then by assuming that \eqref{eq:induction form} holds for some $1\leq k-1 < S$ we have
\begin{align*}
\prod_{s=1}^{k} A_{s} 
&= \big(\Id_{2^{S-k}} \otimes \ones_{1/2} \otimes \Id_{2^{k-1}} \otimes \ones_{1/M}\big)\big(\Id_{2^{S-k+1}} \otimes \ones_{1/2^{k-1}} \otimes \ones_{1/M}\big)\\
&= \big(\big(\Id_{2^{S-k}} \otimes \ones_{1/2}\big)\Id_{2^{S-k+1}}\big)\otimes\big(\big(\Id_{2^{k-1}} \otimes \ones_{1/M}\big)\big(\ones_{1/2^{k-1}} \otimes\ones_{1/M}\big) \big)\\
&=\big(\Id_{2^{S-k}} \otimes \ones_{1/2}\big) \otimes \big(\big(\Id_{2^{k-1}}\ones_{1/2^{k-1}}\big)\otimes \big(\ones_{1/M}\ones_{1/M}\big)\big)\\
&=\big(\Id_{2^{S-k}} \otimes \ones_{1/2}\big) \otimes \big(\ones_{1/2^{k-1}} \otimes \ones_{1/M}\big) \\
&= \Id_{2^{S-k}} \otimes \ones_{1/2^{k}} \otimes\ones_{1/M},
\end{align*}
where the 2nd and 3rd equalities follow from the mixed product property of the Kronecker product. 
\end{proof}

\begin{proof}[Proof of Lemma \ref{lem:facts_about_Vs}]
From \eqref{eq:V_defn_proofs} we have $V_0^i=g(\xi_{0}^i)$ and a proof by induction shows that for $1\leq s \leq S$,
\begin{equation}
V_s^{i_s}=\sum_{(i_0,\ldots,i_{s-1})}g(\xi_{0}^{i_0})\prod_{q=1}^s A_q^{i_qi_{q-1}},\label{eq:V_k_unwind}
\end{equation}
from which \eqref{it:measurability of V} follows. Since $A_{s}$ is row-stochastic, \eqref{it:boundedness of V} follows from \eqref{eq:V_defn_proofs}. In the case $s=S$, \eqref{eq:V_k_unwind} together with Lemma \ref{lem:ones product} gives
\begin{equation*}
V_S^{i_S}=\sum_{(i_0,...,i_{S-1})}\!\!g(\xi_{0}^{i_0})\prod_{q=1}^S A_q^{i_qi_{q-1}}=\sum_{i_0=1}^{N}g(\xi_{0}^{i_0})\Bigg[\prod_{q=1}^S A_q\Bigg]^{i_Si_0}\!\!=\frac{1}{N}\sum_{i_0=1}^{N}g(\xi_{0}^{i_0}).
\end{equation*}
\end{proof}

\begin{proof}[Proof of Proposition \ref{prop:martingale decomposition}]
By the definitions of $X_{\rho}$ and $\cF_{\rho}$ we have \eqref{it:measurability} by Lemma \ref{lem:facts_about_Vs}\eqref{it:measurability of V}. Claim \eqref{it:zero mean} follows from the one step conditional independence and \eqref{eq:P_in_terms_of_V_proofs}. Claim \eqref{it:boundedness} follows from Lemma \ref{lem:facts_about_Vs}\eqref{it:boundedness of V}, \eqref{eq:V_defn_proofs}, and the row-stochasticity of $A_{s}$ for all $1 \leq s \leq S$. It remains to prove \eqref{eq:first decompo} and \eqref{eq:second decompo}.

Since $N^{-1}\sum_{i=1}^{N}V^{i}_{0} \barphi(\xi^{i}_{0}) = 0$, we have the decomposition
\begin{align}\label{eq:telescoping martingale diff}
\frac{1}{N}\sum_{i=1}^{N} V^{i}_{S}\barphi(\xi^{i}_{S}) &= 
\sum_{s = 1}^{S} \left(
\frac{1}{N}\sum_{i_{s}=1}^{N} V^{i_{s}}_{s}\barphi(\xi^{i_{s}}_{s}) 
- \frac{1}{N}\sum_{i_{s-1}=1}^{N} V^{i_{s-1}}_{s-1}\barphi(\xi^{i_{s-1}}_{s-1})
\right)
\end{align}
Because $A_{s}$ is doubly stochastic we have $\sum_{j=1}^{N} A_{s}^{ji}=1$ and hence
\begin{align*}
&\frac{1}{N}\sum_{i_{s}=1}^{N} V^{i_{s}}_{s}\barphi(\xi^{i_{s}}_{s}) 
- \frac{1}{N}\sum_{i_{s-1}=1}^{N}V^{i_{s-1}}_{s-1}\barphi(\xi^{i_{s-1}}_{s-1}) \\
&= \frac{1}{N}\sum_{i_{s}=1}^{N} V^{i_{s}}_{s}\barphi(\xi^{i_{s}}_{s}) 
- \frac{1}{N}\sum_{j=1}^{N}\sum_{i_{s-1}=1}^{N}A^{ji_{s-1}}_{s}V^{i_{s-1}}_{s-1}\barphi(\xi^{i_{s-1}}_{s-1}) \\
&= \frac{1}{N}\sum_{i_{s}=1}^{N} V^{i_{s}}_{s}\left(\barphi(\xi^{i_{s}}_{s}) - \frac{1}{V^{i_{s}}_{s}}\sum_{i_{s-1}=1}^{N}A^{i_{s}i_{s-1}}_{s}V^{i_{s-1}}_{s-1}\barphi(\xi^{i_{s-1}}_{s-1})\right) \\
&= \sqrt{\frac{S}{N}}\sum_{i=1}^{N} X_{(s-1)N+i}.
\end{align*}
By substituting the last form into \eqref{eq:telescoping martingale diff} we obtain \eqref{eq:first decompo}. Finally, since by Lemma \ref{lem:facts_about_Vs}\eqref{it:constant weights} we have that $V_{S}^{i}$ is independent of $i$, we can prove \eqref{eq:second decompo} by writing
\begin{align*}
&\left(\frac{1}{N}\sum_{i=1}^{N}g(\xi^{i}_{0})\right)\left(\frac{1}{N}\sum_{i=1}^{N}\varphi(\xi^{i}_{S})\right) - \frac{1}{N}\sum_{i=1}^{N} g(\xi^{i}_{0})\varphi(\xi^{i}_{0})\\
&=\frac{1}{N}\sum_{i=1}^{N}V^{i}_{S}\varphi(\xi^{i}_{S})
- \frac{1}{N}\sum_{i=1}^{N} g(\xi^{i}_{0})\varphi(\xi^{i}_{0}) \\
&= \frac{1}{N}\sum_{i=1}^{N} V^{i}_{S}\left(\varphi(\xi^{i}_{S}) - \frac{\sum_{i=1}^{N} g(\xi^{i}_{0})\varphi(\xi^{i}_{0})}{\sum_{i=1}^{N}g(\xi^{i}_{0})}\right) \\
&= \frac{1}{N}\sum_{i=1}^{N} V^{i}_{S}\barphi(\xi^{i}_{S}).
\end{align*}
\end{proof}
%
\section{Proofs for Section \ref{sec:r2pf}}
\label{sec:proofs R2PF}

\begin{proof}[Proof of Lemma \ref{lem:mutation martingale}]
We define $\cM_{\mathrm{M}} \defeq \{(X_{\rho},\cA_{\rho}); 1\leq \rho \leq N\}$ as
\begin{align*}
X_{\rho} \defeq \frac{1}{N}\sum_{i=1}^{\rho} \left(\varphi\left(\zeta^{i}\right) - f(\varphi)({\widehat{\zeta}}^{i})\right), \quad
\cA_{\rho} \defeq \sigma\left(\widehat{\zeta}^{1},\ldots,\widehat{\zeta}^{N},\zeta^{1},\ldots,\zeta^{\rho}\right).
\end{align*}
Clearly, for all $1\leq \rho \leq N$, $X_{\rho}$ is $\cA_{\rho}$-measurable, $|X_{\rho}| \leq \infty$, and, by \eqref{eq:joint law}, $\E[X_{\rho}\,|\,\cA_{\rho'}] = X_{\rho'}$ for any $1\leq \rho,\rho' \leq N$ such that $\rho'<\rho$. Hence $\cM_{\mathrm{M}}$ is a martingale and 
\begin{equation*}
X_{N} = \frac{1}{N}\sum_{i=1}^{N}\varphi(\zeta^{i}) - \frac{1}{N}\sum_{i=1}^{N}f(\varphi)(\widehat{\zeta}^{i}).
\end{equation*}
The claim then follows from Burkholder-Davis-Gundy inequality.
\end{proof}

\begin{proof}[Proof of Proposition \ref{prop:induction step}]
Throughout the proof we assume $\varphi \in \boundMeas_{1}(\X)$. To prove part \ref{it:step part a}), we have by Minkowski's inequality
\begin{align}
\lpnorm{r}{\widehat{\pi}^{N}_{n}(\varphi) - \widehat{\pi}_{n}(\varphi)} &\leq \frac{1}{\pi_{n}(g_{n})}\lpnorm{r}{\pi^{N}_{n}(g_{n})\widehat{\pi}^{N}_{n}(\varphi)-\pi^{N}_{n}(g_{n}\varphi)} \nonumber\\&+ \lpnorm{r}{\frac{\pi^{N}_{n}(g_{n}\varphi)}{\pi^{N}_{n}(g_{n})} - \frac{\pi_{n}(g_{n}\varphi)}{\pi_{n}(g_{n})}}\nonumber\\
&+ \frac{1}{\pi_{n}(g_{n})}\lpnorm{r}{\widehat{\pi}^{N}_{n}(\barphi_{n})\left (\pi_{n}(g_{n})-\pi^{N}_{n}(g_{n})\right)},\label{eq:third term}
\end{align}
where $\barphi_{n} \defeq \varphi - \pi^{N}_{n}(g_{n}\varphi)/\pi_{n}^{N}(g_{n})$. For the first term on the r.h.s.~we have by Proposition \ref{prop:intro_to_aug_resampling}
\begin{equation}
\frac{1}{\pi_{n}(g_{n})}\lpnorm{r}{\pi^{N}_{n}(g_{n})\widehat{\pi}^{N}_{n}(\varphi)-\pi^{N}_{n}(g_{n}\varphi)} \leq \frac{2B_{r}\infnorm{g_{n}}}{\pi_{n}(g_{n})}\sqrt{\frac{S}{N}}.
\end{equation}
For the second term we have (similarly as in the proof of Lemma 4 in \cite{crisan_et_doucet02})
\begin{align}
\lpnorm{r}{\frac{\pi^{N}_{n}(g_{n}\varphi)}{\pi^{N}_{n}(g_{n})} - \frac{\pi_{n}(g_{n}\varphi)}{\pi_{n}(g_{n})}} &\leq \frac{\infnorm{g_{n}}}{\pi_{n}(g_{n})}\lpnorm{r}{\pi^{N}_{n}\left(\frac{g_{n}}{\infnorm{g_{n}}} \right) -\pi_{n}\left(\frac{g_{n}}{\infnorm{g_{n}}} \right)} \nonumber\\&+\frac{\infnorm{g_{n}}}{\pi_{n}(g_{n})}\lpnorm{r}{\pi^{N}_{n}\left(\frac{g_{n}\varphi}{\infnorm{g_{n}}} \right) -\pi_{n}\left(\frac{g_{n}\varphi}{\infnorm{g_{n}}} \right)} \nonumber\\
&\leq \frac{2\infnorm{g_{n}}C_{n,r}}{\pi_{n}(g_{n})}\sqrt{\frac{S}{N}}\label{eq:middle term}
\end{align}
where the last inequality uses the assumption \eqref{eq:convergence precondition}. For the third term on the r.h.s.~of \eqref{eq:third term} we have by \eqref{eq:convergence precondition}
\begin{equation}
\frac{1}{\pi_{n}(g_{n})}\lpnorm{r}{\widehat{\pi}^{N}_{n}(\barphi_{n})\left (\pi_{n}(g_{n})-\pi^{N}_{n}(g_{n})\right)} \leq \frac{2\infnorm{g_{n}}C_{n,r}}{\pi_{n}(g_{n})}\sqrt{\frac{S}{N}}.\label{eq:third term bound}
\end{equation}
From \eqref{eq:third term}--\eqref{eq:third term bound}, part \ref{it:step part a}) follows with 
$
\widehat{C}_{n,r} \defeq (2B_{r} + 4C_{n,r})\infnorm{g_{n}}/\pi_{n}(g_{n})$.

For part \ref{it:step part b}), the case $n=0$ follows from Lemma \ref{lem:delmoral}. For the case $n>0$ we can write 
\begin{align*}
\E\left[\left|\pi^{N}_{n+1}(\varphi) - \Phi_{n+1}(\pi^{N}_{n})(\varphi)\right|^{r}\right]^{\frac{1}{r}} 
&\leq \E\left[\left|\pi^{N}_{n+1}(\varphi) - \widetilde{\pi}^{N}_{n+1}(\varphi)\right|^{r}\right]^{\frac{1}{r}} \nonumber\\&+ \E\left[\left|\widehat{\pi}^{N}_{n}(f(\varphi)) - \frac{\pi^{N}_{n}(g_{n}f(\varphi))}{\pi^{N}_{n}(g_{n})}\right|^{r}\right]^{\frac{1}{r}}. 
\end{align*}
For the first term on the r.h.s.~we can use Lemma \ref{lem:mutation martingale} to obtain an upper bound
\begin{equation*}\label{eq:mutation burkholder}
\sup_{\varphi\in\boundMeas_{1}(\X)}\lpnorm{r}{\pi^{N}_{n+1}(\varphi) - \widetilde{\pi}^{N}_{n+1}(\varphi)} \leq 2B_{r}\sqrt{\frac{S}{N}}.
\end{equation*}
By \eqref{ass:regularity}, $\infnorm{g_{n}}/g_{n} \leq \delta$ implying $\pi^{N}_{n}(g_{n})/\infnorm{g_{n}} \geq \delta^{-1}$.
Hence, by Proposition \ref{prop:intro_to_aug_resampling} we have for any $\infnorm{\varphi}\leq 1$
\begin{align*}
\E\left[\left|\widehat{\pi}^{N}_{n}(\varphi) - \frac{\pi^{N}_{n}(g_{n}\varphi)}{\pi^{N}_{n}(g_{n})}\right|^{r}\right]^{\frac{1}{r}} 
&\leq\frac{\delta}{\infnorm{g_{n}}}\E\left[\left|\pi^{N}_{n}(g_{n})\widehat{\pi}^{N}_{n}(\varphi) - {\pi^{N}_{n}(g_{n}\varphi)}\right|^{r}\right]^{\frac{1}{r}} \leq 2B_{r}\delta\sqrt{\frac{S}{N}}. 
\end{align*}
Part \ref{it:step part b}) thus holds with $\widehat{C}_{r} \defeq \max\left(C^\ast_{r},2B_{r}(1+\delta)\right)$.
\end{proof}

%
\section{Proofs for Section \ref{sec:another r2 resampler}}
\label{sec:proof for modified algorithm}

\begin{proof}[Proof of Proposition \ref{prop:gmn convergence}]
The proof is similar to that of Proposition \ref{prop:martingale decomposition}. We define a sequence $\widecheck{\cM} \defeq \{(\wcX_{\rho},\wcF_{\rho});0 \leq \rho \leq N\}$ such that $\wcX_{0} \defeq 0$, $\wcF_{0} \defeq \sigma(\vxi_{\mathrm{in}})$ and for all $1 \leq \rho \leq N$  
\begin{align*}
\wcX_{\rho} \defeq \frac{W^{\rho}_{\mathrm{out}}}{\sqrt{N}}\left(\barphi(\xi^{\rho}_{\mathrm{out}}) - \frac{1}{W^{\rho}_{\mathrm{out}}}\sum_{j=1}^{N}A^{\rho j}g(\xi^{j}_{\mathrm{in}})\barphi(\xi^{j}_{\mathrm{in}})\right),\qquad \wcF_{\rho} \defeq \wcF_{\rho-1} \vee \sigma(\xi^{\rho}_{\mathrm{out}})
\end{align*}
where $\barphi(x) \defeq \varphi(x) - {\pi^{N}(g\varphi)}/{\pi^{N}(g)}$. We show that $\widecheck{\cM}$ is a martingale difference.

Clearly $g(\xi^{i}_{\mathrm{in}})$ is $\cF_{0}$-measurable for all $1\leq i \leq N$ and by \eqref{eq:out weight}, also $W^{\rho}_{\mathrm{out}}$ is $\cF_{0}$-measurable for all $1\leq \rho \leq N$. By the definition of $\wcX_{\rho}$ and $\wcF_{\rho}$, $\wcX_{\rho}$ is thus $\wcF_{\rho}$-measurable for all $\rho \geq 0$. The requirement that $\E[\wcX_{\rho} \,\mid\,\wcF_{\rho-1}] \stackrel{\text{a.s.}}{=} 0$ follows from $\xi^{i}_{\mathrm{out}}$ being conditionally independently distributed according to line \ref{it:withinisland cond dist} in Algorithm \ref{alg:grouped mn}, given $\sigma(\vxi_{\mathrm{in}})$. Finally, by \eqref{eq:out weight}, and the fact that $g\in \boundMeas_{+}(\X)$ we have $|\widecheck{X}_{\rho}| \leq {\infnorm{g}}\osc{\varphi}/\sqrt{N}$. From these observations we conclude that $\widecheck{\cM}$ is a martingale difference.

Next we establish the connection between $\widecheck{\cM}$ and the error term in \eqref{eq:gmn error}.
By the double stochasticity of $A$ and the fact that ${N^{-1}}\sum_{i=1}^{N}g(\xi^{i}_{\mathrm{in}})\barphi(\xi^{i}_{\mathrm{in}}) = 0$ we have
\begin{align}
\frac{1}{\sqrt{N}}\sum_{\rho=0}^{N} \wcX_{\rho} 
&= \frac{1}{N}\sum_{i=1}^{N}W^{i}_{\mathrm{out}}\barphi(\xi^{i}_{\mathrm{out}}) - \frac{1}{N}\sum_{i=1}^{N}\sum_{j=1}^{N}A^{ij}g(\xi^{j}_{\mathrm{in}})\barphi(\xi^{j}_{\mathrm{in}}) \nonumber\\&= \frac{1}{N}\sum_{i=1}^{N}W^{i}_{\mathrm{out}}\barphi(\xi^{i}_{\mathrm{out}}) = \pi^{N}(g)\widecheck{\pi}^{N}(\varphi)-\pi^{N}(g\varphi).\nonumber 
\end{align}
where the last equality follows from the fact that, by \eqref{eq:out weight}, $N^{-1}\sum_{i=1}^{N}W^{i}_{\mathrm{out}} =  \pi^{N}(g)$. Part \ref{it:part a}) of the claim then follows by applying Burkholder-Davis-Gundy inequality to the martingale $\sum_{\rho=0}^{N}\wcX_{\rho}$.

For part \ref{it:part b}) we define for all $1\leq k \leq m$, $\widecheck{\cM}_{k} \defeq \{(\wcX^{k}_{\rho},\wcF^{k}_{\rho});0 \leq \rho \leq M\}$ such that $\wcX^{k}_{0} \defeq 0$, $\wcF^{k}_{0} \defeq \sigma(\vxi_{\mathrm{in}})$ and for all $1 \leq \rho \leq M$
\begin{align*}
\wcX^{k}_{\rho} &\defeq \frac{W^{(k-1)M+\rho}_{\mathrm{out}}}{\sqrt{M}}\left(\barphi_{k}\Big(\xi^{(k-1)M+\rho}_{\mathrm{out}}\Big) - \frac{\sum_{j=1}^{N}A^{(k-1)M+\rho,j}g(\xi^{j}_{\mathrm{in}})\barphi_{k}(\xi^{j}_{\mathrm{in}})}{W^{(k-1)M+\rho}_{\mathrm{out}}}\right),\\
\wcF^{k}_{\rho} &\defeq \wcF^{k}_{\rho-1} \vee \sigma(\xi^{(k-1)M+\rho}_{\mathrm{out}}),
\end{align*}
where $\barphi_{k} = \varphi - \pi^{M,k}(g\varphi)/\pi^{M,k}(g)$. Similarly as above, we can check that $\widecheck{\cM}_{k}$ is a martingale difference with terms bounded by $|\widecheck{X}^{k}_{\rho}| \leq {\infnorm{g}}\osc{\varphi}/{\sqrt{M}}$.

By the definition of $A$ and $\barphi_{n}$ we have 
$$\sum_{j=1}^{N}A^{(k-1)M+\rho,j}g(\xi^{j}_{\mathrm{in}})\barphi_{k}(\xi^{j}_{\mathrm{in}})=\frac{1}{M}\sum_{j=1}^{M}g(\xi^{(k-1)M+j}_{\mathrm{in}})\barphi_{k}(\xi^{(k-1)M+j}_{\mathrm{in}}) = 0,$$ 
which together with \eqref{eq:out weight} enables us to write
\begin{align*}
\frac{1}{\sqrt{M}}\sum_{\rho=0}^{M} \wcX^{k}_{\rho} 
= \frac{1}{M}\sum_{\rho=1}^{M}W^{(k-1)M+\rho}_{\mathrm{out}}\barphi_{k}(\xi^{(k-1)M+\rho}_{\mathrm{out}})
= \pi^{M,k}(g)\left(\widecheck{\pi}^{M,k}(\varphi) - \frac{\pi^{M,k}(g\varphi)}{\pi^{M,k}(g)}\right).
\end{align*}
The claim then follows by Burkholder-Davis-Gundy inequality as before.
\end{proof}

\begin{proof}[Proof of Proposition \ref{prop:martingale for parsimonious communication}]
By definition, $\tX_{\rho}$ depends only on $\vxi^{r(\rho)}_{s(\rho)}$, $\vxi^{\ell(\rho)}_{s(\rho)}$, $\bV^{r(\rho)}_{s(\rho)}$, $\bV^{\ell(\rho)}_{s(\rho)}$, as well as $\bV^{i}_{s(\rho)-1}$ and $\vxi^{i}_{s(\rho)-1}$ for all $1\leq i \leq m$. The measurability then follows similarly by the $\sigma(\vxi_{0})$-measurability of $\bV^{i}_{s}$ for all $1\leq s \leq S$ and $1 \leq i \leq m$ as for Proposition \ref{prop:martingale decomposition group} by Lemma \ref{lem:facts_about_Vs}.

Fix $1 \leq \rho \leq  Sm/2$. By definition $$\bV^{\ell(\rho)}_{s(\rho)} = \bV^{r(\rho)}_{s(\rho)} = \frac{1}{2}\left(\bV^{\ell(\rho)}_{s(\rho)-1} + \bV^{r(\rho)}_{s(\rho)-1}\right)$$ and we can write 
\begin{equation}\label{eq:equality or pairs}
\begin{split}
\frac{1}{\bV^{\ell(\rho)}_{s(\rho)}}\sum_{j=1}^{m}\bA_{s(\rho)}^{\ell(\rho)j}\bV^{j}_{s(\rho)-1}\vbarphi(\vxi^{j}_{s(\rho)-1}) &= \frac{1}{\bV^{r(\rho)}_{s(\rho)}}\sum_{j=1}^{m}\bA_{s(\rho)}^{r(\rho)j}\bV^{j}_{s(\rho)-1}\vbarphi(\vxi^{j}_{s(\rho)-1})\\ &=p_{\ell}\vbarphi(\vxi^{\ell(\rho)}_{s(\rho)-1}) + p_{r}\vbarphi(\vxi^{r(\rho)}_{s(\rho)-1})
\end{split}
\end{equation}
where $p_{\ell} \defeq \frac{1}{2}\bV^{\ell(\rho)}_{s(\rho)-1}/\bV^{\ell(\rho)}_{s(\rho)}$ and $p_{r} \defeq \frac{1}{2}\bV^{r(\rho)}_{s(\rho)-1}/\bV^{r(\rho)}_{s(\rho)}$ and hence $p_{\ell} + p_{r} = 1$. By \eqref{eq:paired X} and \eqref{eq:equality or pairs}
\begin{align*}
 \E\left[\tX_{\rho}\,\big|\,\tF_{\rho-1}\right] 
&= \frac{\bV}{\sqrt{Sm}}\left(\E\left[\vbarphi(\vxi^{\ell(\rho)}_{s(\rho)})\,\Big|\,\tF_{\rho-1}\right] + \E\left[\vbarphi(\vxi^{r(\rho)}_{s(\rho)})\,\Big|\,\tF_{\rho-1}\right]\right.\\&\left.\qquad - 2\left(p_{\ell}\vbarphi(\vxi^{\ell(\rho)}_{s(\rho)-1}) + p_{r}\vbarphi(\vxi^{r(\rho)}_{s(\rho)-1})\right)\right),
\end{align*}
where $\bV \defeq \bV^{\ell(\rho)}_{s(\rho)} = \bV^{r(\rho)}_{s(\rho)}$ and the conditional expectations can be written explicitly by observing the conditional probabilities 
\begin{align*}
\P(\vxi^{\ell(\rho)}_{s(\rho)} = \vxi^{\ell(\rho)}_{s(\rho)-1} \, | \, \tF_{\rho-1}) &= p_\ell^2 + 2p_rp_\ell, \qquad
\P(\vxi^{\ell(\rho)}_{s(\rho)} = \vxi^{r(\rho)}_{s(\rho)-1} \, | \, \tF_{\rho-1}) = p_r^2,\\
\P(\vxi^{r(\rho)}_{s(\rho)} = \vxi^{r(\rho)}_{s(\rho)-1} \, | \, \tF_{\rho-1}) &= p_r^2 + 2p_rp_\ell, \qquad
\P(\vxi^{r(\rho)}_{s(\rho)} = \vxi^{\ell(\rho)}_{s(\rho)-1} \, | \, \tF_{\rho-1}) = p_\ell^2,
\end{align*}
and the fact that 
\begin{align*}
\P(\vxi^{\ell(\rho)}_{s(\rho)} \notin \{\vxi^{\ell(\rho)}_{s(\rho)-1},\vxi^{r(\rho)}_{s(\rho)-1}&\} \, | \, \tF_{\rho-1}) = \P(\vxi^{r(\rho)}_{s(\rho)} \notin \{\vxi^{\ell(\rho)}_{s(\rho)-1},\vxi^{r(\rho)}_{s(\rho)-1}\} \, | \, \tF_{\rho-1}) = 0.
\end{align*}
The claim $\E[\tX_{\rho}|\tF_{\rho-1}] = 0$ of part (b) then follows by straightforward calculation.

Part (c) follows from the boundedness of $g$ by the definition of $\tX_{\rho}$ and part (d) follows from observing that
\begin{equation*}
\sqrt{\frac{S}{m}}\sum_{\rho=1}^{Sm/2} \tX_{\rho} = \sqrt{\frac{S}{m}}\sum_{s=1}^{S}\sum_{i=1}^{m}\tX^{i}_{s}
\end{equation*}
where $\tX^{i(\rho)}_{s(\rho)}$ is defined in an otherwise identical manner as $\bX_{\rho}$, except for the law of $\vxi^{i(\rho)}_{s(\rho)}$ being different for $\tX^{i(\rho)}_{s(\rho)}$ and $\bX_{\rho}$, which does not affect the validity of \eqref{eq:first decompo modified} and \eqref{eq:second decompo modified}.
%
\end{proof}

%
\section{Proofs for Section \ref{sec:grouped r2pf}}
\label{sec:proof for airpf}

\begin{proof}[Proof of Propostition \ref{prop:gr2ph induction step 1}]

We start the proof by introducing two more empirical approximations of $\widecheck{\pi}_{n}$ 
\begin{equation}\label{eq:def after within island rs}
\widecheck{\pi}^{M,k}_{n} \defeq \dfrac{\sum_{i=1}^{M}\widecheck{W}^{(k-1)M+i}_{n}\delta_{\widecheck{\zeta}_{n}^{(k-1)M+i}}}{\sum_{i=1}^{M}\widecheck{W}^{(k-1)M+i}_{n}},
\quad
\widecheck{\pi}^{N}_{n} \defeq \dfrac{\sum_{i=1}^{N}\widecheck{W}^{i}_{n}\delta_{\widecheck{\zeta}_{n}^{i}}}{\sum_{i=1}^{N}\widecheck{W}^{i}_{n}}
\quad 1\leq k \leq m,
\end{equation}
based on the samples $\widecheck{\vzeta}^{k}_{n} \defeq (\widecheck{\zeta}^{(k-1)M+1}_{n},\ldots,\widecheck{\zeta}_{n}^{kM})$ and $\widecheck{\vzeta}_{n} \defeq (\widecheck{\zeta}^{1}_{n},\ldots,\widecheck{\zeta}_{n}^{N})$, respectively. Note that $\widecheck{\pi}^{N}_{n}$ is the approximation of $\widehat{\pi}_{n}$ based on the entire sample after the first within island resampling while $\widecheck{\pi}^{M,k}_{n}$ is the PE specific approximation based on the sample contained in $k$th PE.

By Minkowski's inequality we have
\begin{align}\label{eq:threeway decompo}
\begin{split}
\lpnorm{r}{\widehat{\pi}^{M,k}_{n}(\varphi) - \widehat{\pi}_{n}(\varphi)} &\leq \lpnorm{r}{\widehat{\pi}^{M,k}_{n}(\varphi) - \widecheck{\pi}^{M,k}_{n}(\varphi)} \\&+ \lpnorm{r}{\widecheck{\pi}^{M,k}_{n}(\varphi)-\frac{\pi^{M,k}_{n}(g_{n}\varphi)}{\pi^{M,k}_{n}(g_{n})}} \\&+ \lpnorm{r}{\frac{\pi^{M,k}_{n}(g_{n}\varphi)}{\pi^{M,k}_{n}(g_{n})} -\frac{\pi_{n}(g_{n}\varphi)}{\pi_{n}(g_{n})}}
\end{split}
\end{align}

For the last term on the r.h.s.~we have by using \eqref{eq:gr2pf induction ass}, similarly as in equation \eqref{eq:middle term} in the proof of Proposition \ref{prop:induction step}, 
\begin{equation}\label{eq:example of crisan}
\lpnorm{r}{\pi^{M,k}_{n}(g_{n}\varphi)/\pi^{M,k}_{n}(g_{n}) -\pi_{n}(g_{n}\varphi)/\pi_{n}(g_{n})} \leq \frac{2\infnorm{g_{n}}C_{n,r}}{\pi_{n}(g_{n})\sqrt{M}}.
\end{equation}

The remainder of the proof hinges on expressing the first two terms in the r.h.s.~of \eqref{eq:threeway decompo} in terms of expectations such as the one in \eqref{eq:example of crisan} as well as expectations of the form $\E[|\widecheck{\pi}^{M,i}_{n}(\barphi_{n}^{i})|^{r}]$ where $\barphi_{n}^{i} \defeq \varphi - \pi^{M,i}_{n}(g_{n}\varphi)/\pi^{M,i}_{n}(g_{n})$ and $\E[|\widecheck{\pi}^{N}_{n}(\barphi_n)|^{r}]$ where $\barphi_{n} \defeq \varphi - \pi^{N}_{n}(g_{n}\varphi)/\pi^{N}_{n}(g_{n})$. 
The last two types of expectations can  be bounded by using the identities
\begin{align}\label{eq:indentities}
\begin{split}
\widecheck{\pi}^{M,i}_{n}(\barphi_{n}) &= 
\frac{{\pi}^{M,i}_{n}(g_{n})\Big(\widecheck{\pi}^{M,i}_{n}(\varphi)-\frac{\pi^{M,i}_{n}(g_{n}\varphi)}{\pi^{M,i}_{n}(g_{n})}\Big)}{\pi_{n}(g_{n})} + \frac{\widecheck{\pi}^{M,i}_{n}(\barphi_{n})}{\pi_{n}(g_{n})}\Big(\pi_{n}(g_{n})-{\pi}^{M,i}_{n}(g_{n})\Big) \\
\widecheck{\pi}^{N}_{n}(\barphi_{n}) &= 
\frac{{\pi}^{N}_{n}(g_{n})\widecheck{\pi}^{N}_{n}(\varphi)-\pi^{N}_{n}(g_{n}\varphi)}{\pi_{n}(g_{n})} + \frac{\widecheck{\pi}^{N}_{n}(\barphi_{n})}{\pi_{n}(g_{n})}\left(\pi_{n}(g_{n})-{\pi}^{N}_{n}(g_{n})\right)
\end{split}
\end{align}
together with Proposition \ref{prop:gmn convergence} and assumption \eqref{eq:gr2pf induction ass}.

Clearly, the second term on the r.h.s.~of \eqref{eq:threeway decompo} is readily of the desired form and therefore we only need to consider the first term. Again, by applying Minkowski's inequality, we have
\begin{align}\label{eq:first partial}
\begin{split}
\lp{r}{\widehat{\pi}^{M,k}_{n}(\varphi) - \widecheck{\pi}^{M,k}_{n}(\varphi)} 
&\leq 
\lp{r}{\widehat{\pi}^{M,k}_{n}(\varphi) - \widecheck{\pi}^{N}_{n}(\varphi)} + \lp{r}{\widecheck{\pi}^{N}_{n}(\barphi_{n})} \\&+
\lp{r}{\widecheck{\pi}^{M,k}_{n}(\barphi_{n}^{k})} \\ &+ \lp{r}{{\pi^{M,k}_{n}(g_{n}\varphi)}/{\pi^{M,k}_{n}(g_{n})}-{\pi^{N}_{n}(g_{n}\varphi)}/{\pi^{N}_{n}(g_{n})}} 
\end{split}
\end{align}
If we write $P_{i} \defeq \P(\widehat{\pi}_{n}^{M,k} = \widecheck{\pi}_{n}^{M,i}\,|\,\widecheck{\vzeta}_{n},\,{\vzeta}_{n})$, where $1\leq i \leq m$, then by the tower property of conditional expectations, Cauchy-Schwartz, Jensen's and Minkowski's inequalities we have
\begin{align}
 \lpnorm{r}{\widehat{\pi}^{M,k}_{n}(\varphi) - \widecheck{\pi}^{N}_{n}(\varphi)} 
&=\E\left[\sum_{i=1}^{m} P_{i}\left|\widecheck{\pi}^{M,i}_{n}(\varphi) - \widecheck{\pi}^{N}_{n}(\varphi)\right|^{r}\right]^{\frac{1}{r}} \nonumber\\
&\leq \sum_{i=1}^{m}\E\left[\left|\widecheck{\pi}^{M,i}_{n}(\varphi) - \widecheck{\pi}^{N}_{n}(\varphi)\right|^{2r}\right]^{\frac{1}{2r}} \nonumber\\
&\leq  m\lpnorm{2r}{\widecheck{\pi}^{N}_{n}(\barphi_{n})} + \sum_{i=1}^{m} \lpnorm{2r}{\widecheck{\pi}^{M,i}_{n}(\barphi_{n}^{i})} \nonumber\\&+ \sum_{i=1}^{m}\lpnorm{2r}{\frac{\pi^{M,i}_{n}(g_{n}\varphi)}{\pi^{M,i}_{n}(g_{n})}-\frac{\pi^{N}_{n}(g_{n}\varphi)}{\pi^{N}_{n}(g_{n})}} \label{eq:second partial}
\end{align}
In \eqref{eq:first partial} and \eqref{eq:second partial} all terms are of the desired form except for the terms $\lp{r}{{\pi^{M,i}_{n}(g_{n}\varphi)}/{\pi^{M,i}_{n}(g_{n})}-{\pi^{N}_{n}(g_{n}\varphi)}/{\pi^{N}_{n}(g_{n})}}$ where $1 \leq i \leq m$, but for these terms we have for all $1 \leq i \leq m$
\begin{align*}
\lpnorm{r}{\frac{\pi^{M,i}_{n}(g_{n}\varphi)}{\pi^{M,i}_{n}(g_{n})}-\frac{\pi^{N}_{n}(g_{n}\varphi)}{\pi^{N}_{n}(g_{n})}} &\leq \lpnorm{r}{\frac{\pi^{M,i}_{n}(g_{n}\varphi)}{\pi^{M,i}_{n}(g_{n})}-\frac{\pi_{n}(g_{n}\varphi)}{\pi_{n}(g_{n})}} \\&+ \lpnorm{r}{\frac{\pi^{N}_{n}(g_{n}\varphi)}{\pi^{N}_{n}(g_{n})}-\frac{\pi_{n}(g_{n}\varphi)}{\pi_{n}(g_{n})}} \leq \frac{4\infnorm{g_{n}}C_{n,r}}{\pi_{n}(g_{n})\sqrt{M}}
\end{align*}
where the final inequality follows similarly as in \eqref{eq:example of crisan} by using \eqref{eq:gr2pf induction ass}. By applying \eqref{eq:indentities} together with Proposition \ref{prop:gmn convergence} and \eqref{eq:gr2pf induction ass} then claim then follows with
\begin{equation*}
\widehat{C}_{n,r} = 2\left(\left(1+\frac{1}{\sqrt{m}}\right)mB_{2r} + \left(2+\frac{1}{\sqrt{m}}\right)B_{r}+6C_{n,r}+4mC_{n,2r}\right)\frac{\infnorm{g_{n}}}{\pi_{n}(g_{n})}.
\end{equation*}
\end{proof}

\begin{proof}[Proof of Propostion \ref{prop:induction step M fixed}]
By Minkowski's inequality we have
\begin{align}
\lpnorm{r}{\widehat{\pi}^{N}_{n}(\varphi)-\widehat{\pi}_{n}(\varphi)} &\leq 
\lpnorm{r}{\widehat{\pi}^{N}_{n}(\varphi)-\widecheck{\pi}^{N}_{n}(\varphi)} \nonumber\\&+
\lpnorm{r}{\widecheck{\pi}^{N}_{n}(\varphi) - \frac{\pi^{N}_{n}(g_{n}\varphi)}{\pi^{N}_{n}(g_{n})}} \nonumber\\&+
\lpnorm{r}{\frac{\pi^{N}_{n}(g_{n}\varphi)}{\pi^{N}_{n}(g_{n})} - \frac{\pi_{n}(g_{n}\varphi)}{\pi_{n}(g_{n})}},\label{eq:another threeway decompo}
\end{align}
For the third term on the r.h.s.~of \eqref{eq:another threeway decompo} we have similarly as in the proof of Proposition \ref{prop:gr2ph induction step 1}
\begin{equation}
\lp{r}{\pi^{N}_{n}(g_{n}\varphi)/\pi^{N}_{n}(g_{n}) - \pi_{n}(g_{n}\varphi)/\pi_{n}(g_{n})} \leq \frac{2\infnorm{g_{n}}C_{n,r}}{\pi_{n}(g_{n})}\sqrt{\frac{S}{m}}.
\end{equation}
It remains to consider the first two terms on the r.h.s.~of \eqref{eq:another threeway decompo}.

Let us write, analogously to Proposition \ref{prop:g augmented convergence},
$$\vg_{n}(\vx) \defeq \frac{1}{M}\sum_{j=1}^{M} g_{n}(x^i),$$
for all $\vx = (x^{1},\ldots,x^{M}) \in \X^{M}$. By the definition of the weights $(\widecheck{W}_{n}^{1},\ldots,\widecheck{W}_{n}^{N})$ in \eqref{eq:out weight} we see that for all $1 \leq i \leq m $, $\widecheck{W}^{(i-1)M+1}_{n}=\vg_{n}(\widecheck{\vzeta}^{i}_{n}) =\vg_{n}({\vzeta}^{i}_{n}) $ and $\sum_{i=1}^{N} \widecheck{W}^{i}_{n} = \sum_{i=1}^{N} g_{n}(\zeta^{i}_{n})$. This enables us to write 
\begin{align*}
\widecheck{\pi}^{N}_{n}(\varphi) 
= \dfrac{\frac{1}{m}\sum_{i=1}^{m}\widecheck{W}^{(i-1)M+1}_n\frac{1}{M}\sum_{j=1}^{M}\varphi(\widecheck{\zeta}^{(i-1)M+j}_{n})}{\frac{1}{N}\sum_{i=1}^{N}g_{n}(\zeta^{i}_{n})} = \frac{1}{{\pi^{N}_{n}(g_{n})}}\frac{1}{m}\sum_{i=1}^{m}\vg_{n}(\widecheck{\vzeta}^{i}_{n})\vphi(\widecheck{\vzeta}^{i}_{n}),
\end{align*}
yielding the identity
\begin{align*}
\widehat{\pi}^{N}_{n}(\varphi) - \widecheck{\pi}^{N}_{n}(\varphi)  &= 
\frac{1}{\pi_{n}(g_{n})}\left(\pi^{N}_{n}(g_{n})\widehat{\pi}^{N}_{n}(\varphi) - \frac{1}{m}\sum_{i=1}^{m}\vg_{n}(\widecheck{\vzeta}^{i}_{n})\vphi(\widecheck{\vzeta}^{i}_{n})\right)  \\&+  \frac{1}{\pi_{n}(g_{n})}\left(\widecheck{\pi}^{N}_{n}(\varphi) - \widehat{\pi}^{N}_{n}(\varphi) \right)\left(\pi^{N}_{n}(g_{n}) - \pi_{n}(g_{n})\right),
\end{align*}
which, together with Proposition \ref{prop:g augmented convergence} and \eqref{eq:group convergence precondition} yields
\begin{align*}
\lp{r}{\widehat{\pi}^{N}_{n}(\varphi)-\widecheck{\pi}^{N}_{n}(\varphi)} 
&\leq 2\left(B_{r}+C_{n,r}\right)\frac{\infnorm{g_{n}}}{\pi_{n}(g_{n})}\sqrt{\frac{S}{m}}.
\end{align*}
\noindent
For the second term on the r.h.s.~of \eqref{eq:another threeway decompo} we have, similarly as in the proof of Proposition \ref{prop:gr2ph induction step 1}, by using Proposition \ref{prop:g augmented convergence}
\begin{equation}
\lpnorm{r}{\widecheck{\pi}^{N}_{n}(\varphi) - \pi^{N}_{n}(g_{n}\varphi)/\pi^{N}_{n}(g_{n})} \leq \left(\frac{B_{r}}{\sqrt{M}} + C_{n,r}\right)\frac{2\infnorm{g_{n}}}{\pi_{n}(g_{n})}\sqrt{\frac{S}{m}}.
\end{equation}
The claim then follows with $\widehat{C}_{n,r} = ((2+2/\sqrt{M})B_{r}+6C_{n,r})\infnorm{g_n}/\pi_{n}(g_n)$.
\end{proof}

\end{appendices}

%
%

\end{document}